\documentclass[final]{fac}
\pdfoutput=1
\usepackage[T1]{fontenc}
\usepackage[utf8]{inputenc}

\usepackage{microtype}
\usepackage{xcolor}
\usepackage{rotating}
\usepackage{ifdraft}
\usepackage{subcaption}
\usepackage[noend]{algpseudocode}
\usepackage{mathtools}
\usepackage{float}
\usepackage{wrapfig}
\usepackage{nicefrac}
\usepackage{booktabs}
\usepackage{multirow}
\usepackage{bm}
\usepackage{etoolbox}
\usepackage{dsfont}
\usepackage{enumitem}
\usepackage{amssymb,amsmath}
\usepackage{thmtools}
\usepackage{pifont}
\usepackage{xspace,relsize,stmaryrd}
\usepackage[final]{listings}
\usepackage{tikz}
\usetikzlibrary{cd, calc, shapes, automata, fit, decorations.pathreplacing}

\usepackage{hyperref}

\input{tikz-functor-composition}

\tikzstyle{scissor lines}=[
    linestyle/.style={
      dashed,
      draw=black!50,
    },
    scissors/.style={
      text=black!50,
      font=\small,
      inner sep=0pt,
    },
    gapstyle/.style={
      draw=white,
      line width=1mm,
    }
]

\tikzset{shiftarr/.style={
        rounded corners,%
        to path={--([#1]\tikztostart.center)
                     -- ([#1]\tikztotarget.center) \tikztonodes
                     -- (\tikztotarget)},
}}

\tikzset{horizontal then diagonal/.style={
        rounded corners,%
        to path={
          -- ($ (\tikztostart.center) !.5! (\tikztostart.center -| \tikztotarget.center) $) \tikztonodes
          -- (\tikztotarget)
                     },
}}

\usepackage[author=anonymous,marginclue,footnote]{fixme}
\FXRegisterAuthor{ls}{als}{LS}%
\FXRegisterAuthor{sm}{asm}{SM}%
\FXRegisterAuthor{tw}{atw}{TW}%
\FXRegisterAuthor{hp}{ahp}{HP}%
\DeclareUnicodeCharacter{D7}{\ensuremath{\times}}

\newcommand{\makecell}[2][]{\begin{tikzpicture}[baseline=(X.base)]
    \node[inner sep=0pt,outer sep=0pt,draw=none,align=left,
    execute at begin node=\setlength{\baselineskip}{3ex}] (X) {#2};
    \end{tikzpicture}%
}

\newcommand{\scissors}{\ding{34}} %

\usepackage{seqsplit}
\usepackage{xstring}

\makeatletter
\usepackage[notcite,notref]{showkeys}%

\newcommand\noshowkeys{\def\hideNextShowKeysLabel{test}}
\renewcommand*\showkeyslabelformat[1]{%
\@ifundefined{hideNextShowKeysLabel}{%
\noexpandarg%
\StrSubstitute{#1}{ }{\textvisiblespace}[\TEMP]%
\parbox[t]{\marginparwidth}{\raggedright\normalfont\small\ttfamily\(\{\){\color{red!50!black}\expandafter\seqsplit\expandafter{\TEMP}}\(\}\)}%
}{}%
}

\makeatother

\newcommand{\copar}{\textsf{CoPaR}\xspace}
\newcommand{\berkeleyparser}{\textsf{berkeleyparser}\xspace}

\newcommand\vartextvisiblespace[1][.5em]{%
  \makebox[#1]{%
    \kern.07em
    \vrule height.3ex
    \hrulefill
    \vrule height.3ex
    \kern.07em
  }
}

\newcommand{\id}{\mathsf{id}}
\newcommand{\inl}{\mathsf{inl}}
\newcommand{\inr}{\mathsf{inr}}

\newcommand{\Set}{\ensuremath{\mathsf{Set}}\xspace}

\newcommand{\Pow}{{\mathcal{P}}}
\newcommand{\Powf}{{\mathcal{P}_{\!\omega}}}
\newcommand{\Bag}{\Bagf}
\newcommand{\Bagf}{\ensuremath{\mathcal{B}_{\omega}}\xspace}
\newcommand{\BagM}{\ensuremath{\mathcal{B}(M_{\neq 0})}}
\newcommand{\Dist}{\ensuremath{{\mathcal{D}_\omega}}\xspace}
\newcommand{\R}{\mathds{R}}
\newcommand{\N}{\mathds{N}}
\newcommand{\Z}{\mathds{Z}}

\newcommand{\CO}{\mathcal{O}}

\newcommand{\op}[1]{\ensuremath{\mathsf{#1}}}

\newcommand{\fpair}[1]{\langle#1\rangle}

\newcommand{\rifactor}{\ensuremath{p}} %

\newcommand{\sortprod}[1][i]{{\textstyle\prod_{#1\in I}}}
\newcommand{\sortprodImpl}[1][i]{{\textstyle\prod}}
\newcommand{\Potf}{\Powf} %
\newcommand{\inj}{\ensuremath{\mathsf{in}}}
\newcommand{\pr}{\ensuremath{\mathsf{pr}}}
\newcommand{\sortcoprod}[1][i]{{\textstyle\coprod_{#1\in I}}}

\newcommand{\sortcoprodImpl}[1][i]{{\textstyle\coprod}}
\newcommand{\arity}[3][]{\ensuremath{\mathord{\raisebox{1pt}{\ensuremath{#1#2}}\mkern-1.5mu/\mkern-1.5mu{\raisebox{-1pt}{\ensuremath{#1#3}}}}}}

\newcommand{\descto}[3][]{\arrow[phantom]{#2}[#1]{\text{\footnotesize{}\begin{tabular}{c}#3\end{tabular}}}}

\newcommand{\itemref}[2]{\autoref{#1}.\ref{#2}}

\newcommand{\lmbrace}{\{\mskip-4mu[}
\newcommand{\rmbrace}{]\mskip-4mu\}}
\newcommand{\mbraces}[1]{\lmbrace{}\,#1\,\rmbrace{}}

\tikzset{
  coalgebra drawing/.style={
    state/.append style={
      minimum width=0pt,
      minimum height=0pt,
      inner sep=0.8mm,
    },
    every edge/.append style={
      shorten <= 1pt,
      shorten >= 1pt,
    }
  }
}
\tikzstyle{searchtree}=[
    level distance = 5mm,
    every node/.style={
        inner sep=1pt,
        minimum width=4mm,
        minimum height=4mm,
        draw=black,
        shape=circle,
        outer sep=0pt,
    },
    subtree/.append style={
      minimum height=9mm,
    },
    noedge/.style={
        edge from parent/.style={}
    },
    child anchor=north,
    level 1/.style={sibling distance=0mm, level distance=5mm}, %
    level 2/.style={sibling distance=11mm, level distance=4mm},
    level 3/.style={sibling distance=9mm, level distance=6mm}, 
]

\tikzstyle{root}=[
draw=none,
minimum width=1mm,
minimum height=1mm,
]
\tikzstyle{subtree}=[
  isosceles triangle,
  draw,
  shape border rotate=90,
  isosceles triangle stretches=true,
  isosceles triangle apex angle=50,
  minimum height=8mm,
  minimum width=8mm,
  inner sep=2pt,
  yshift={0mm},
  anchor=north,
]

\tikzstyle{only math nodes}=[%
  every node/.append style={%
   execute at begin node=$,%
   execute at end node=$%
  }%
]

\def\epito{\twoheadrightarrow}
\newcommand{\csum}{\raisebox{-1pt}{\text{\large$\mathrm{\Sigma}$}}}
\newcommand{\takeout}[1]{\empty}
\newcommand{\geZero}{>^{\!\!\smash{?}}}

\numberwithin{equation}{section}

\setlist[enumerate,1]{label=(\arabic*),font=\normalfont,align=left,leftmargin=0pt,labelindent=0pt,listparindent=\parindent,labelwidth=0pt,itemindent=!,topsep=3pt,parsep=0pt,itemsep=3pt,start=1}
\setlist[enumerate,2]{label=(\alph*),font=\normalfont,labelindent=*,leftmargin=*,start=1}

\newtheorem{theorem}{Theorem}[section]
\newtheorem{definition}[theorem]{Definition}
\newtheorem{notation}[theorem]{Notation}
\newtheorem{example}[theorem]{Example}
\newtheorem{proposition}[theorem]{Proposition}
\newtheorem{corollary}[theorem]{Corollary}
\newtheorem{remark}[theorem]{Remark}
\newtheorem{construction}[theorem]{Construction}

\title{From Generic Partition Refinement to Weighted Tree Automata Minimization}

\author[T.~Wißmann, H.P.~Deifel, S.~Milius, L.~Schröder]
{
 Thorsten Wißmann\texorpdfstring{\thanks{Supported by the DFG project COAX (MI 717/5-2)}$^,$\thanks{Supported by the DFG project COAX (SCHR 1118/12-2)},
  Hans-Peter Deifel\texorpdfstring{\thanks{Supported by the Deutsche Forschungsgemeinschaft (DFG) as part of the Research and Training Group 2475 ``Cybercrime and Forensic Computing'' (393541319/GRK2475/1-2019)}}{},
 Stefan Milius\texorpdfstring{\footnotemark[1]$^,$\footnotemark[3]}{}, and
 Lutz Schröder\texorpdfstring{\footnotemark[2]$^,$\footnotemark[3]}{}\\
 Friedrich-Alexander-Universit\"{a}t Erlangen-N\"urnberg, Germany}{}}

\correspond{Thorsten Wißmann <thorsten.wissmann@fau.de>}

\hypersetup{final,hidelinks}
\hypersetup{
  bookmarks=true,
  colorlinks=true,
  linkcolor=black,
  citecolor=black,
  filecolor=black,
  urlcolor=black,
  pdftitle={\csname @title\endcsname},
  pdfauthor={Thorsten Wißmann, Hans-Peter Deifel, Stefan Milius, Lutz Schröder},
  pdfduplex={DuplexFlipLongEdge},
}

\pubyear{2020}
\pagerange{\pageref{firstpage}--\pageref{lastpage}}

\begin{document}
\makecorrespond

\maketitle
\label{firstpage}

\begin{abstract} Partition refinement is a method for minimizing
  automata and transition systems of various types. Recently, we have
  developed a partition refinement algorithm that is generic in the
  transition type of the given system and matches the run time of the
  best known algorithms for many concrete types of systems,
  e.g.~deterministic automata as well as ordinary, weighted, and
  probabilistic (labelled) transition systems. Genericity is achieved
  by modelling transition types as functors on sets, and systems
  as coalgebras. In the present work, we refine the run time analysis
  of our algorithm to cover additional instances, notably weighted
  automata and, more generally, weighted tree automata.  For weights
  in a cancellative monoid we match, and for non-cancellative monoids
  such as (the additive monoid of) the tropical semiring even
  substantially improve, the asymptotic run time of the best known
  algorithms. We have implemented our algorithm in a generic tool that
  is easily instantiated to concrete system types by implementing a
  simple refinement interface. Moreover, the algorithm and the tool
  are modular, and partition refiners for new types of systems are
  obtained easily by composing pre-implemented basic
  functors. Experiments show that even for complex system types, the
  tool is able to handle systems with millions of transitions.
\end{abstract}
\section{Introduction}
Minimization is a basic verification task on state-based systems, concerned with
reducing the number of system states as far as possible while preserving the
system behaviour. This can be done by identifying states that exhibit the same
behaviour. Hence, it can be used for equivalence checking of systems, and
constitutes a preprocessing step in further system analysis tasks, such as model
checking.

Notions of equivalent behaviour typically vary quite widely even on fixed system
types~\cite{vanglabbeek2001linear}. We work with various notions of
\emph{bisimilarity}, i.e.\ with branching-time equivalences. Classically,
bisimilarity for labelled transition systems obeys the principle ``states $x$
and $y$ are bisimilar if for every transition $x \to x'$, there exists a
transition $y \to y'$ with $x'$ and $y'$ bisimilar, and vice versa''%
.
It is thus given via a fixpoint
definition, to be understood as a \emph{greatest} fixpoint, and can
therefore be iteratively approximated from above. This is the
principle behind \emph{partition refinement} algorithms: Initially all
states are tentatively considered equivalent, and then this initial
partition is iteratively refined according to observations made on the
states until a fixpoint is reached. Unsurprisingly, such procedures
run in polynomial time.  Its comparative tractability (in contrast,
e.g.~trace equivalence and language equivalence of non-deterministic
systems are PSPACE-complete~\cite{KanellakisS90}) makes minimization
under bisimilarity interesting even in cases where the main
equivalence of interest is linear-time, such as word automata.\twnote{}

Kanellakis and Smolka~\cite{KanellakisS90} in fact provide a minimization
algorithm with run time $\CO(m\cdot n)$ for ordinary transition systems with $n$
states and $m$ transitions. However, even faster partition refinement algorithms
running in $\CO((m+n)\cdot \log n)$ have been developed for various types of
systems over the past 50 years. For example, Hopcroft's algorithm minimizes
deterministic automata for a fixed input alphabet $A$ in $\CO(n\cdot \log n)$
\cite{Hopcroft71}; it was later generalized to variable input alphabets, with
run time $\CO(n\cdot |A|\cdot \log n)$~\cite{Gries1973,Knuutila2001}. The
Paige-Tarjan algorithm minimizes transition systems in time
$\CO((m+n)\cdot \log n)$~\cite{PaigeTarjan87}, and generalizations to labelled
transition systems have the same time
complexity~\cite{HuynhTian92,DerisaviEA03,Valmari09}. Minimization of weighted
systems is typically called \emph{lumping} in the literature; Valmari and
Franchescinis~\cite{ValmariF10} exhibit a simple $\CO((m+n)\cdot \log n)$
lumping algorithm for systems with rational weights.

In earlier work~\cite{concur2017,concurSpecialIssue}\smnote{} we have
developed an efficient \emph{generic} partition refinement algorithm that can be
easily instantiated to a wide range of system types, most of the time either
matching or improving the previous best run time. The genericity of the
algorithm is based on modelling state-based systems as coalgebras following the
paradigm of universal coalgebra~\cite{Rutten00}, in which the branching
structure of systems is encapsulated in the choice of a functor, the \emph{type
  functor}. This allows us to cover not only classical relational systems and
various forms of weighted systems, but also to combine existing system types in
various ways, e.g.~nondeterministic and probabilistic branching. Our algorithm
uses a functor-specific \emph{refinement interface} that supports a graph-based
representation of coalgebras. It allows for a generic complexity analysis, and
indeed the generic algorithm has the same asymptotic complexity as the
above-mentioned specific algorithms. For Segala
systems~\cite{Segala95} (systems that combine probabilistic and
non-deterministic branching, also known as Markov decision processes), it
matches the run time of a recent algorithm~\cite{GrooteEA18} discovered
independently and almost at the same time as ours, and improves on the run time
of the previously best algorithm~\cite{BaierEM00}.

The new contributions of the present paper are twofold. On the theoretical side,
we show how to instantiate our generic algorithm to weighted systems with
weights in a monoid (generalizing the group-weighted case considered
previously~\cite{concur2017,concurSpecialIssue}). We then refine the complexity
analysis of the algorithm, making the complexity of the implementation of the
type functor a parameter $\rifactor(c)$, where $c$ is the input coalgebra. In
the new setup, the previous analysis becomes the special case where
$\rifactor(c) = 1$. Under the same structural assumptions on the type functor
and the refinement interface as previously, our algorithm runs in time $\CO(m
\cdot \log n\cdot \rifactor(c))$ for an input coalgebra $c$ with $n$ states and
$m$ transitions. Instantiated to the case of weighted systems over
non-cancellative monoids (with $\rifactor(c) = \log(m)$ where $m$ is the number
of transitions in $c$) the run time of the generic algorithm is
$\CO(m\cdot\log^2 m)$, thus markedly improving the run time $\CO(m\cdot n)$ of
previous algorithms for weighted automata~\cite{Buchholz08} and, more generally,
(bottom-up) weighted tree automata~\cite{HoegbergEA07}. This includes weighted
tree automata for the additive monoid $(\N,\max,0)$ of the tropical semiring,
which are used in natural language processing~\cite{MayKnight06}.
In addition, for cancellative monoids, we again essentially match the complexity of
the previous algorithms~\cite{Buchholz08,HoegbergEA07}.

Our second main contribution is a generic and modular implementation
of our algorithm, the \emph{Coalgebraic Partition Refiner}
(\copar). Instantiating \copar to coalgebras for a given functor
requires only to implement the refinement interface. We provide such
implementations for a number of basic type functors, e.g.\ for
non-deterministic, weighted, or probabilistic branching, as well as
(ranked) input and output alphabets or output weights. In addition,
\copar is \emph{modular}: For any type functor obtained by composing
basic type functors for which a refinement interface is available,
\copar automatically derives an implementation of the refinement
interface. We explain in detail how this modularity is realized in our
implementation and, extending Valmari and Franchescinis's
ideas~\cite{ValmariF10}, we explain how the necessary data structures
need to be implemented so as to realize the low theoretical
complexity. We thus provide a working efficient partition refiner for
all the above mentioned system types. In particular, our tool is, to
the best of our knowledge, the only available implementation of
partition refinement for many composite system types, notably for
weighted (tree) automata over non-cancellative monoids. The tool
including source code and evaluation data is available at
\url{https://git8.cs.fau.de/software/copar}.

The present paper is an extended and completely reworked version of a previous
conference paper~\cite{coparFM19}. It includes full proofs, additional
benchmarks, and more extensive examples and explanations. Moreover, we formally
show how refinement interfaces can be combined along products of functors
(\autoref{prodInterface} and \autoref{sec:des}). We have optimized the memory
consumption of our implementation which has led to better performance in the
benchmarks on weighted tree automata (\autoref{tab:extended}).

\paragraph{Organization.} The material is structured as follows.  In
\autoref{sec:theofoundations} we recall the necessary technical background and
the modelling of state based systems as coalgebras. In \autoref{sec:partref}, we
describe the tool and the underlying algorithm, discussing in particular tool
usage and implementation, the generic interface, and the modularity principles
that we employ. Some concrete instantiations are exhibited in
\autoref{sec:instances}. We then go on to elaborate the case of weighted systems
in more detail, giving a refinement interface for the basic underlying functor
of such systems in \autoref{sec:monoidvalued}, and showing in \autoref{sec:wta}
how to cover weighted tree automata -- which arise by combining weighted
systems and ranked alphabets -- by means of our modularity principles.
Benchmarks are presented in \autoref{sec:bench}.

\section{Preliminaries: Universal Coalgebra}

\label{sec:theofoundations}
Our algorithmic framework~\cite{concurSpecialIssue} is based on modelling
state-based systems abstractly as \emph{coalgebras} for a (set) \emph{functor}
that encapsulates the transition type, following the paradigm of \emph{universal
  coalgebra}~\cite{Rutten00}. We proceed to recall standard notation for sets
and maps, as well as basic notions and examples in coalgebra. Occasional
comments assume familiarity with basic notions of category theory
(e.g.~\cite{Awodey10}) but the few concepts needed for the main development are
explained in full. We fix a singleton set $1=\{*\}$; for every set~$X$ we have a
unique map $!\colon X\to 1$. We denote composition of maps by $(-)\cdot(-)$, in
applicative order. We denote the disjoint union -- in categorical terms, the
\emph{coproduct} -- of sets~$A,B$ by $A+B$ where we write $\inl\colon A\to A+B$
and $\inr\colon B\to A+B$ for the canonical injections; the disjoint union, or
coproduct, of a family $(X_j)_{j\in J}$ of sets is denoted by
$\coprod_{j \in J} X_j$. Similarly, we write $\prod_{j\in J} X_j$ for the
(cartesian) product of a family of sets. Injection maps of disjoint unions and
projection maps of products, respectively, are denoted by
\[
  \inj_i\colon X_i \to \coprod_{j \in J} X_j
  \qquad
  \text{and}
  \qquad
  \pr_i\colon \prod_{j\in J} X_j \to X_i.
\]
Given two maps $f\colon A\to X$ and $g\colon A\to Y$ we write
$\fpair{f,g}\colon A\to X\times Y$ for the map $a \mapsto
(f(a),g(a))$. Similarly, for a family of maps $(f_i\colon A\to X_i)_{i\in I}$,
we write $\fpair{f_i}_{i\in I}\colon A\to \sortprod[i] X_i$ for the map
$a\mapsto (f_i(a))_{i\in I}$.

We model the transition type of state based systems using \emph{functors}.
Informally, a functor~$F$ assigns to a set~$X$ a set~$FX$, whose elements are
thought of as structured collections over~$X$, and an $F$-coalgebra is a
map~$c\colon X\to FX$ assigning to each state~$x$ in a system a structured
collection $c(x)\in FX$ of successors. The most basic example is that of
transition systems, where~$F$ is powerset, so a coalgebra assigns to each state
a set of successors. Formal definitions are as follows.

\begin{definition}
  \begin{enumerate}
  \item A \emph{functor} $F\colon \Set\to\Set$ assigns to each set~$X$ a
    set~$FX$, and to each map $f\colon X\to Y$ a map $Ff\colon FX\to FY$,
    preserving identities and composition ($F\id_X=\id_{FX}$,
    $F(g\cdot f)=Fg\cdot Ff$).
  \item An \emph{$F$-coalgebra} $(C,c)$ consists of a set~$C$ of \emph{states}
    and a \emph{transition} structure $c\colon C\to FC$.
  \item A \emph{morphism} $h\colon (C,c)\to (D,d)$ of $F$-coalgebras is a map
    $h\colon C\to D$ that preserves the transition structure,
    i.e.~$Fh\cdot c = d\cdot h$.
  \item Two states $x,y\in C$ of a coalgebra $c\colon C\to FC$ are
    \emph{behaviourally equivalent} (notation: $x \sim y$) if there exists a
    coalgebra morphism $h$ such that $h(x) = h(y)$.
  \end{enumerate}
\end{definition}
As above, we usually use the letters $X$ and $Y$ for sets (without structure)
and $C$ or $D$ for state sets of coalgebras.
\begin{example} \label{exManyFunctors}
  \begin{enumerate}
  \item \label{exManyFunctors:Pow} The \emph{finite powerset} functor $\Powf$
    maps a set~$X$ to the set~$\Powf X$ of all \emph{finite} subsets of~$X$, and
    a map $f\colon X\to Y$ to the map $\Powf f = f[-]\colon \Powf X\to \Powf Y$
    taking direct images.  $\Powf$-coalgebras are finitely branching
    (unlabelled) transition systems and two states are behaviourally equivalent
    iff they are bisimilar in the sense of Milner~\cite{Milner80} and
    Park~\cite{Park81}.

  \item A signature $\Sigma$ is a set $\Sigma$ of \emph{operation symbols}
    together with a map $\op{ar}\colon \Sigma\to \N$, which assigns to each
    operation symbol $\sigma \in \Sigma$ its \emph{arity} $\op{ar}(\sigma)$.
    We write $\arity{\sigma}{n}\in \Sigma$ for $\sigma \in \Sigma$
    with $\op{ar}(\sigma) = n$. Every signature $\Sigma$ canonically defines a
    \emph{polynomial functor}
    \[
      F_{\Sigma} X = \coprod_{\arity[\scriptstyle]{\sigma}{n} \in \Sigma} X^{n}.
    \]
    We slightly abuse notation by denoting for each
    $\arity{\sigma}{n} \in \Sigma$ the corresponding injection into the coproduct by
    \[
      \sigma\colon X^n\to F_{\Sigma} X.
    \]
    Moreover, we simply write $\Sigma$ in lieu of $F_{\Sigma}$, so we have
    \[
      \Sigma X = F_\Sigma X = \{\sigma(x_1,\ldots,x_n)\mid \arity{\sigma}{n} \in \Sigma, 
      x_1,\ldots,x_n\in X\};
    \]
    for $\sigma(x_1,\ldots,x_n)$, we sometimes write
    $\arity{\sigma}{n}(x_1,\ldots,x_n)$ to emphasize the arity or disambiguate
    overloaded symbols. This polynomial functor acts component-wise on maps
    $f\colon X\to Y$:
    \[
      \Sigma f\colon \Sigma X\to \Sigma Y
      \qquad
      (\Sigma f)(\sigma(x_1,\ldots,x_n)) = \sigma(f(x_1),\ldots,f(x_n))
      \qquad\text{for $\arity{\sigma}{n} \in \Sigma$}.
    \]
    Every state in a $\Sigma$-coalgebra represents a (possibly infinite)
    \emph{$\Sigma$-tree}, i.e.\ a rooted ordered tree where every node is
    labelled with some operation symbol $\sigma \in \Sigma$ and has precisely
    $\op{ar}(\sigma)$-many children. In particular, a node is a leaf iff it is
    labelled with a $0$-ary operation symbol. For example, for the signature
    $\Sigma = \{\arity{*}{2}, \arity{\ell}{0}\}$ with a binary operation symbol
    and a constant, we have the following example of a $\Sigma$-tree:
    \[
      \takeout{} %
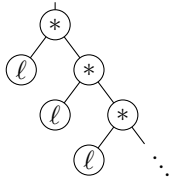
      \begin{tikzpicture}[searchtree,only math nodes,
        child anchor=border,
        level 2/.style={sibling distance=9mm, level distance=6mm},
        level 3/.style={sibling distance=9mm, level distance=6mm}, 
        level 4/.style={sibling distance=9mm, level distance=6mm}, 
        ]
        \node[draw=none] (pseudoroot) {}
        child{
          node{*}
          child{node{\ell}}
          child{node{*}
            child{node{\ell}}
            child{node{*}
              child{node{\ell}}
              child{node[draw=none,rotate=35,inner sep=0pt]{\vdots}}
            }
          }
        }
        ;
      \end{tikzpicture}
    \]
    Given a state $x$ in a coalgebra $c\colon C \to \Sigma C$, we obtain a
    $\Sigma$-tree $t_x$ by unravelling the coalgebra structure at~$x$. More
    precisely, $t_x$ is uniquely defined by
    \[
      t_x =
      \takeout{} %
      \begin{tikzpicture}[searchtree,only math nodes,
        level 1/.style={sibling distance=6mm, level distance=5mm},
        baseline=(dots.base)
        ]
        \node{\sigma}
        child{node[subtree]{t_{x_1}}}
        child{
          node[yshift=-4mm,draw=none] (dots) {\cdots}
          edge from parent[draw=none]}
        child{node[subtree]{t_{x_{\mathrlap{n}\phantom{1}}}}};
      \end{tikzpicture}
      \qquad
      \text{if $c(x) = \sigma(x_1, \ldots, x_n)$}.
    \]
    (this equation constituting a \emph{coinductive}
    definition~\cite{Rutten00}).  For example, the above $\Sigma$-tree is
    obtained by unravelling the coalgebra structure at the state $x$ of the
    $\Sigma$-coalgebra
    \[
      c\colon \{x,y\} \to \Sigma \{x,y\}
      \quad\text{with}\quad \text{$c(x) = *(y, x)$ and $c(y) = \ell$}.
    \]
    Two states in a $\Sigma$-coalgebra are behaviourally equivalent iff they
    represent the same possibly infinite tree: $x \sim y$ iff $t_x = t_y$.
    
  \item For a fixed finite set $A$, the functor given by
    $FX=2\times X^A$, where $2 = \{0,1\}$, sends a set $X$ to the set
    of pairs of boolean values and functions $A\to X$. An
    $F$-coalgebra $(C,c)$ is a deterministic automaton (without
    initial state). For each state $x\in C$, the first component of
    $c(x)$ determines whether $x$ is a final state, and the second
    component is the successor function $A\to X$ mapping each input
    letter $a\in A$ to the successor state of $x$ under input letter
    $a$. States $x,y \in C$ are behaviourally equivalent iff they
    accept the same language in the usual sense.

    This functor is (naturally isomorphic to) the polynomial functor for the
    signature $\Sigma$ consisting of two operation symbols of arity $|A|$:
    $2 \times X^A \cong X^{|A|} + X^{|A|}$.
  \item\label{exManyFunctors:3} For a commutative monoid $(M,+,0)$, the
    \emph{monoid-valued} functor $M^{(-)}$ sends each set $X$ to the set of
    \emph{finitely supported} maps $f\colon X\to M$, i.e.~$f(x) = 0$ for all but
    finitely many $x\in X$. In case~$M$ is even an abelian group, we also refer
    to $M^{(-)}$ as a \emph{group-valued} functor.

    An $F$-coalgebra $c\colon C\to M^{(C)}$ is,
    equivalently, a finitely branching $M$-weighted transition system: For a
    state $x\in C$, $c(x)$ maps each state $y\in C$ to the weight $c(x)(y)$ of
    the transition from $x$ to $y$. For a map $f\colon X\to Y$,
    $M^{(f)}\colon M^{(X)}\to M^{(Y)}$ sends a finitely supported map
    $v\colon X\to M$ to the map $y\mapsto \sum_{x\in X, f(x) = y} v(x)$,
    corresponding to the standard image measure construction. As the notion of
    behavioural equivalence of states in $M^{(-)}$-coalgebras, we obtain
    weighted bisimilarity (cf.~\cite{Buchholz08,KlinS13}), given coinductively
    by postulating that states $x,y\in C$ are behaviourally equivalent
    $(x\sim y)$ iff
      \[\textstyle
        \sum_{z'\sim z} c(x)(z')
        = \sum_{z'\sim z} c(y)(z')\qquad \text{for all $z\in C$}.
      \]
      For the Boolean monoid $(2=\{0,1\},\vee,0)$, the monoid-valued
      functor $2^{(-)}$ is (naturally isomorphic to) the finite
      powerset functor $\Powf$.  For the monoid of real numbers
      $(\R,+,0)$, the monoid-valued functor $\R^{(-)}$ has
      $\R$-weighted systems as coalgebras, e.g.~Markov chains. In
      fact, finite Markov chains are precisely finite coalgebras of
      the \emph{finite distribution functor}, i.e.~the
      subfunctor~$\Dist$ of $\R_{\ge 0}^{(-)}$ (and hence
      of~$\R^{(-)}$) given by
      \[
        \Dist(X) = \{\mu \in \R_{\ge 0}^{(X)}\mid\textstyle{\sum_{x \in X} \mu(x)
        = 1} \}.
      \]
      For the monoid $(\N,+,0)$ of natural numbers, the monoid-valued
      functor is the bag functor $\Bag$, which maps a set $X$ to the
      set of finite multisets over~$X$.
  \end{enumerate}
\end{example}
\begin{notation}\label{N:csum}
  Note that for every commutative monoid $(M,+,0)$, we have the canonical
    summation map
    \begin{equation}\label{eq:csum}
      \csum\colon \Bag M\to M
      \qquad\text{with}\qquad
      \csum f = \sum_{x\in M} \underbrace{x+\cdots+ x}_{f(x)\text{-many}}
    \end{equation}
    It sums up all elements of a bag $f$ of monoid
    elements, where a single element of the monoid can occur multiple
    times.
\end{notation}
  
\begin{remark}\label{R:EilenbergMoore}
  For categorically-minded readers, we note that $\Bag$ is a monad on the
  category of sets. Moreover, commutative monoids are precisely the
  Eilenberg-Moore algebras (e.g.~\cite{Awodey10}) for $\Bag$. In fact, for every
  commutative monoid $(M,+,0)$, the map $\csum$ is the structure of its
  associated Eilenberg-Moore algebra.

    \takeout{}%
\end{remark}

\section{Generic Partition Refinement}
\label{sec:partref}

We recall some key aspects of our generic partition refinement
algorithm~\cite{concurSpecialIssue}, which \emph{minimizes} a
given coalgebra, i.e.~computes its quotient modulo behavioural
equivalence. We centre the presentation around the implementation and
use of our tool.

The algorithm~\cite[Algorithm 4.5]{concurSpecialIssue} is parametrized over a
type functor $F$, represented by implementing a fixed \emph{refinement
  interface}, which in particular allows for a representation of $F$-coalgebras
in terms of nodes and edges (by no means implying a restriction to relational
systems!). Our previous analysis has established that the algorithm minimizes an
$F$-coalgebra $c\colon C\to FC$ with~$n$ nodes and~$m$ edges in time $\CO(m\cdot
\log n)$, assuming $m\ge n$ and that the operations of the refinement interface
run in linear time. In the present paper, we generalize the analysis,
establishing a run time in $\CO(m\cdot \log n\cdot \rifactor(c))$, where
$\rifactor(c)$ is a factor in the time complexity of the operations implementing
the refinement interface which depends on the input coalgebra $c\colon C\to FC$.
For many functors, $\rifactor(c) = 1$, reproducing the previous analysis. In some
cases, $\rifactor(c)$ is not constant, and our new analysis still applies in
these cases, either matching or improving the best known run time complexity in most
instances, most notably weighted systems over non-cancellative monoids.

We proceed to discuss the design of the implementation, including input formats
of our tool \copar for composite functors built from pre-implemented basic
blocks and for systems to be minimized (\autoref{sec:parsing}). We then discuss
the internal representation of coalgebras in the tool
(\autoref{sec:coalgint}). Subsequently, we recall refinement interfaces,
describe their implementation (\autoref{sec:refintface}), and discuss how to
combine them (\autoref{sec:des}). Finally, we note implementation
details of our tool and, in particular, argue that it realizes the theoretical
time complexity (\autoref{sec:imp}). 

\subsection{Generic System Specification}
\label{sec:parsing}
\copar accepts as input a file that represents a finite $F$-coalgebra
$c\colon C\to FC$, and consists of two parts. The first part is a
single line specifying the functor $F$. Each of the remaining lines
describes one state $x\in C$ and its one-step behaviour
$c(x)$. Examples of input files are shown in
\autoref{fig:example-input}.

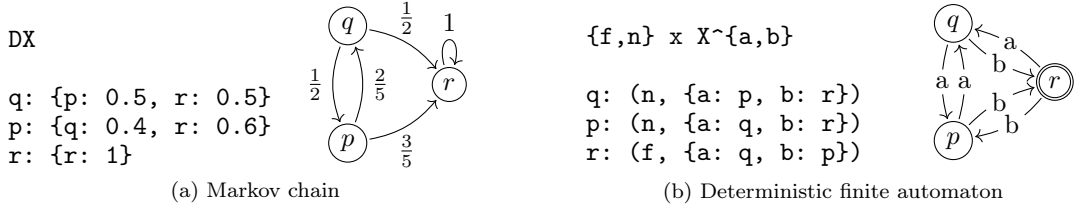
\begin{figure}
  \centering%
  \vspace{-2mm}
  \hfill
  \begin{subfigure}{0.4\textwidth}%
    \begin{minipage}[b]{.55\textwidth}%
\begin{verbatim}
DX

q: {p: 0.5, r: 0.5}
p: {q: 0.4, r: 0.6}
r: {r: 1}
\end{verbatim}
    \end{minipage}
    \hfill%
    \begin{tikzpicture}[coalgebra drawing]
      \node (q1) at (120:9mm) [state] {$q$};
      \node (q2) at (-120:9mm) [state] {$p$};
      \node (q3) at (0:9mm) [state] {$r$};
      \path[->,bend angle=20]
      (q1) edge [left, bend right] node[overlay] {$\frac{1}{2}$} (q2)
      (q1) edge [above,bend left] node[overlay] {$\frac{1}{2}$} (q3)
      (q2) edge [right, bend right] node {$\frac{2}{5}$} (q1)
      (q2) edge [below, bend right,overlay] node {$\frac{3}{5}$} (q3)
      (q3) edge [loop above] node {$1$} (q3);
    \end{tikzpicture}%
    \hfill\hspace*{0pt}%
    \caption{Markov chain}
    \label{subfig:markov}
  \end{subfigure}%
  \hfill%
  \begin{subfigure}{0.4\textwidth}
    \begin{minipage}[b]{.58\textwidth}
\begin{verbatim}
{f,n} x X^{a,b}

q: (n, {a: p, b: r})
p: (n, {a: q, b: r})
r: (f, {a: q, b: p})
\end{verbatim}
    \end{minipage}%
    \hfill%
    \begin{tikzpicture}[coalgebra drawing]
      \node (q) at (120:9mm) [state] {$q$};
      \node (p) at (-120:9mm) [state] {$p$};
      \node (r) at (0:9mm) [state,accepting] {$r$};
      \path[->,bend angle=15,
      every node/.append style={
        shape=circle,
        inner sep=1pt,
        fill=white,
        anchor=center,
      }]
      (q) edge [bend right] node {a} (p)
      (q) edge [bend right] node {b} (r)
      (p) edge [bend right] node {a} (q)
      (p) edge [bend left] node {b} (r)
      (r) edge [bend right] node {a} (q)
      (r) edge [bend left = 20] node {b} (p);
    \end{tikzpicture}
    \hfill%
    \caption{Deterministic finite automaton}
    \noshowkeys\label{fig:dfa-picture}
  \end{subfigure}\hfill\,
  \caption{Examples of input files with encoded coalgebras}
  \vspace{-4mm}
  \label{fig:example-input}
\end{figure}

\subsubsection{Functor Specification}
Functors are specified as composites of basic building blocks; that is, the
functor given in the first line of an input file is an expression determined by
the grammar
\begin{equation}
T ::= \texttt{X} \mid F(T,\ldots,T) \qquad (F\colon \Set^k\to \Set) \in \mathcal{F},
\label{termGrammar}
\end{equation}
where the character $\texttt{X}$ is a terminal symbol and $\mathcal{F}$ is a set
of predefined symbols called \emph{basic functors}, representing a number of
pre-implemented functors of type $F\colon \Set^k\to \Set$. Only for basic
functors, a \emph{refinement interface} needs to be implemented
(\autoref{sec:refintface}); for composite functors, the tool derives an
appropriate refinement interface automatically (\autoref{sec:des}). Basic
functors currently implemented include the (finite) powerset functor~$\Powf$,
the bag functor~$\Bag$, monoid-valued functors~$M^{(-)}$, and polynomial
functors for finite many-sorted signatures $\Sigma$, based on the description of
the respective refinement interfaces given in our previous
work~\cite{concurSpecialIssue} and, in the case of $M^{(-)}$ for
unrestricted commutative monoids~$M$ (rather than only abelian groups), the newly
developed interface described in \autoref{sec:noncan}. Since behavioural
equivalence is preserved and reflected under converting $G$-coalgebras into
$F$-coalgebras for a subfunctor~$G$
of~$F$~\cite[Proposition~2.13]{concurSpecialIssue}, we also cover subfunctors,
such as the finite distribution functor~$\Dist$ as a subfunctor of $\R^{(-)}$.
With the polynomial constructs $+$ and $\times$ written in infix notation as
usual, the currently supported grammar is effectively
\begin{align}
  T &::= \texttt{X}
      \mid
      \Powf\, T \mid \Bag\, T \mid \Dist\, T
          \mid M^{(T)} 
      \mid \Sigma \label{termGrammar2}
\\
  \Sigma &::= C \mid T + T \mid T \times T \mid T^A
      \quad~
  C ::= \N ~\vert~ A
      \quad~
  A ::= \{s_1,\ldots,s_n\} \mid n\notag
\end{align}
where $n\in\N$ denotes the set $\{0,\ldots,n-1\}$, the $s_k$ are strings
subject to the usual conventions for variable names\footnote{a letter or an
  underscore character followed by alphanumeric characters or underscore},
exponents $F^A$ are written \verb|F^A|, and
$M$ is one of the monoids $(\Z,+,0)$, $(\R,+,0)$, $(\mathds{C},+,0)$,
$(\Powf(64), \cup, \emptyset)$ (i.e.\ the monoid of $64$-bit words
with bitwise $\op{or}$),
and $(\N,\max,0)$ (the additive monoid of the tropical semiring). Note
that~$C$ effectively ranges over at most countable sets, and~$A$ over
finite sets. A term~$T$ determines a functor $F\colon \Set\to\Set$ in
the evident way, with~$\texttt{X}$ interpreted as the argument,
i.e.~$F(\texttt{X}) = T$. It should be noted that the implementation
treats composites of polynomial (sub-)terms as a single functor in
order to minimize overhead incurred by excessive decomposition,
e.g.~$X \mapsto \{a,b\} + \Powf(\R^{(X)}) + X \times X$ is composed
of the basic functors $\Powf$, $\R^{(-)}$ and the $3$-sorted
polynomial functor $\Sigma(X,Y,Z) = \{a,b\} + X + Y \times Z$.

\subsubsection{Coalgebra Specification}
The remaining lines of an input file define a finite $F$-coalgebra
$c\colon C\to FC$. Each line of the form $x\texttt{:}\text{\textvisiblespace} t$
defines a state $x\in C$, where $x$ is a variable name, and $t$ represents the
element $t=c(x)\in FC$. The syntax for $t$ depends on the specified functor~$F$,
and follows the structure of the term~$T$ defining~$F$; we write $t\in T$ for a
term $t$ describing an element of~$FC$:
\begin{itemize}
\item $t\in\texttt{X}$ is given by one of the named
  states specified in the file.
\item $t\in T_1\times\dots\times T_n$ is given by
  $t ::= (t_1,\dots,t_n)$ where $t_i \in T_i$, $i=1,\dots, n$.
\item $t\in T_1 + \dots + T_n$ is given by
  $t ::= \texttt{inj}\text{\textvisiblespace}
  i\text{\textvisiblespace} t_i$ where $i=1,\dots, n$ and $t_i \in T_i$.
\item $t\in \Powf T $ and $t\in\Bag T$ are given by
  $t ::= \texttt{\{}t_1,\ldots,t_n\texttt{\}}$ with
  $t_1,\dots,t_n \in T$.
\item $t\in M^{(T)}$ is given by
  $t::=\texttt{\{}t_1\texttt{:}\text{\textvisiblespace}m_1\texttt{,}
  \ldots\texttt{,}$ $t_n\texttt{:}\text{\textvisiblespace}m_n \texttt{\}}$ with
  $m_1,\dots,m_n\in M$ and $t_1,\dots,t_n\in T$, denoting $\mu \in M^{(TC)}$
  with $\mu(t_i) = m_i$ for $i=1,\dots n$, and $\mu(t)=0$ for
  $t\notin\{t_1,\dots,t_n\}$.
\end{itemize}
For example,  the two-line declaration
\begin{verbatim}
P({a,b} x R^(X))
x: {(a, {x: 2.4}), (a,{}), (b,{x: -8})}
\end{verbatim}
defines an $F$-coalgebra
for the functor $FX=\Powf(\{a,b\}\times\R^{(X)}$),
with a single state~$x$, having two
$a$-successors and one~$b$-successor, where successors are elements of
$\R^{(X)}$. One $a$-successor is constantly zero, and the other assigns
weight~$2.4$ to~$x$; the $b$-successor assigns weight $-8$ to~$x$. Two more
examples are shown in Fig.~\ref{fig:example-input}.
 
\subsubsection{Generic Input File Processing}
After reading the functor term $T$, the tool builds a parser for the
functor-specific input format and parses an input coalgebra specified in the
above syntax into an intermediate format described in the next section.
In the case of a composite functor, the parsed coalgebra then undergoes a
substantial amount of preprocessing that also affects how transitions are
counted; we defer the discussion of this point to \autoref{sec:des}, and assume
for the time being that $F\colon \Set\to\Set$ is a basic functor with only one
argument.

\subsection{Internal Representation of Coalgebras}
\label{sec:coalgint}

New functors are added to the framework by implementing a \emph{refinement
  interface} (\autoref{D:refinement-interface}). The interface relates to an
abstract encoding of the functor and its coalgebras in terms of nodes and edges:
\removebrackets
\begin{definition}[\cite{concurSpecialIssue}]\label{D:enc}
  An \emph{encoding} of a functor $F$ consists of a set $A$ of
  \emph{labels} and a family of maps
  \[
    \flat\colon FX\to \Bag(A\times X),
  \]
  one for every set~$X$. The
  \emph{encoding} of an $F$-coalgebra $c\colon C\to FC$ is given by
  the map
  \[
    \fpair{F!,\flat}\cdot c\colon C\to F1\times \Bag(A\times C),
  \] and we
  say that the coalgebra has $n = |C|$ states and
  $m = \sum_{x\in C}|\flat(c(x))|$ edges.
\end{definition}
\noindent
An encoding does by no means imply a reduction from $F$-coalgebras to
$\Bag(A\times (-))$-coalgebras, i.e.\ the notions of behavioural equivalence for
$\Bag(A\times (-))$ and~$F$, respectively, can be radically different. The
encoding just fixes a representation format.
\begin{remark}
  Categorically-minded readers will notice that $\flat$ is not assumed to be a
  natural transformation. In fact,~$\flat$ fails to be natural in all encodings
  we have implemented except the one for polynomial functors.
\end{remark}
Encodings typically match how one intuitively draws coalgebras of various types
as certain labelled graphs. We briefly recall three examples below;
see~\cite{concurSpecialIssue} for more. We note that so far, we see no general
method for deriving an encoding of a functor, which therefore requires
invention.

\begin{example}\label{E:enc}
  \begin{enumerate}
  \item\label{E:enc:1} We have mentioned in
    \autoref{exManyFunctors}\ref{exManyFunctors:Pow} that finitely branching
    transition systems are the coalgebras for $F= \Powf$. For the encoding we
    choose the singleton set $A = 1$ of labels, and
    $\flat\colon \Powf X \to \Bag(1 \times X) \cong \Bag X$ is the obvious
    inclusion, i.e.\ $\flat(S)(x) = 1$ if $x \in S$ and $\flat(S)(x) = 0$
    otherwise, for $S\in\Powf X$.
  \item\label{E:enc:2} For a monoid-valued functor $F = M^{(-)}$ (see
    \autoref{exManyFunctors}\ref{exManyFunctors:3}) we take $A = M_{\neq 0}$,
    the non-zero elements of~$M$, and define
    $\flat\colon M^{(X)} \to \Bag(M_{\neq 0} \times X)$ by taking $\flat(f)$ to be the finite
    set $\{(f(x),x) \mid x \in X, f(x) \neq 0\}$, interpreted as a bag.

    Special cases are the group-valued functors $G^{(-)}$ for an abelian
    group~$G$, in particular $\R^{(-)}$ and its subfunctor $\Dist$, whose
    coalgebras are Markov chains (cf.~Fig.~\ref{fig:example-input}). In the last
    case, we formally inherit $A = \R_{\neq 0}$ from the encoding of $\R^{(-)}$
    but can actually restrict to $A=(0,1]$.
  \item\label{E:enc:3} For a polynomial functor $F = \Sigma$, the set of labels is $A= \N$,
    and the map $\flat\colon \Sigma X \to \Bag(\N \times X)$ is given by 
    \[
      \flat(\sigma(x_1,\ldots,x_n)) = \big\{
      (1,x_1),
      \ldots,
      (n,x_n)
      \big\}.
    \]
  \end{enumerate}
\end{example}
\takeout{}%
\noindent The implementation of a basic functor then consists of two
ingredients: (1)~a parser that transforms the syntactic specification of an
input coalgebra (\autoref{sec:parsing}) into the encoded coalgebra in the above
sense, and (2)~an implementation of the refinement interface, which is
motivated next.

\subsection{Splitting Blocks by $F3$}
\label{sec:f3}
In order to understand the requirements on an interface encapsulating the
functor specific parts of partition refinement, let us look at one step of the
algorithm which is crucial for the overall run time complexity. Partition
refinement algorithms in general maintain a
\emph{partition} of the state space, i.e.\ a disjoint decomposition of the state
space into sets called \emph{blocks}, adhering to the invariant that states in
different blocks are behaviourally inequivalent, and ensuring upon termination
that states in the same block are behaviourally equivalent. Initially, the
algorithm tentatively identifies all states of a coalgebra $c\colon C\to FC$ in
a partition consisting of only one block, $C$. Then, the algorithm splits this
block into smaller blocks whenever states of the coalgebra turn out to be
behaviourally inequivalent and successively applies this procedure to the new
blocks until no further splitting is necessary. In the first iteration, the
algorithm separates states $x,y\in C$ if they are distinguished by
$F!\cdot c\colon C\to F1$, i.e.~if $F!(c(x)) \neq F!(c(y))$.  For example for
automata, i.e.~for $FX= 2 \times X^A$, we have $F1 = 2\times 1^A \cong 2$, so this
first step separates final from non-final states. In the classical Paige-Tarjan
algorithm~\cite{PaigeTarjan87}, i.e.~for $FX = \Powf X$, deadlock states and
states with at least one outgoing transition are separated from each other. In
the subsequent steps, the representation of the coalgebra as labelled edges
(i.e.~$\flat\cdot c\colon C\to \Bagf(A\times C)$) is used to refine the
partition further. Information about the inequivalence of states is propagated from
successor states to predecessor states; this is iterated until a (greatest)
fixed point is reached, i.e.~until no new behavioural inequivalences are
discovered.

\begin{figure} \centering
  \tikzstyle{partitionBlock}=[%
        shape=rectangle,
        rounded corners=2.5mm,
        minimum height=5mm,
        minimum width=5mm,
        draw=black!40,
        inner sep = 2mm,
]%
\begin{tikzpicture}
  \begin{scope}[
    every node/.append style={
      inner sep=0pt,
      outer sep=2pt,
      minimum width=1pt,
      minimum height=4pt,
      anchor=center,
    }
    ]
    \begin{scope}[yshift=1.5cm]
      \node (x1) at (-1.2,0) {$y$};
      \node (x2) at (0,0) {$x$};
      \node (x3) at (1.2,0) {$z$};
    \end{scope}
    \node (y1) at (-1.6,0) {$\bullet$};
    \node (y2) at (-0.5,0) {$\bullet$};
    \node (y3) at (0.5,0) {$\bullet$};
    \node (y4) at (1.2,0) {$\bullet$};
  \end{scope}
  \node[anchor=west,xshift=2mm] (y5) at (y4.east) {$\ldots$};
  \draw (y1) edge[draw=none] node {$\ldots$} (y2);
  \begin{scope}[
    every node/.append style={
      partitionBlock,
      inner sep=2pt,
    },
    every label/.append style={
      font=\footnotesize,
      anchor=south,
      outer sep=1pt,
      inner sep=1pt,
      minimum height=2mm,
      shape=rectangle,
    },
    ]
    \node[fit=(x1) (x2) (x3)] (x123){};
    \node[fit=(y1) (y5)] (C) {};
    \node[every label] at ([xshift=-1mm]C.north east) {\normalsize $B$};
    \begin{scope}[every node/.append style={minimum width=10mm}, 
      ]
      \newcommand{\nextblockdistance}{8mm} 
      \foreach \name/\nodesource/\anchor/\direction in
      {Peast/x123/east/,
        Qeast/C/east/,
        Pwest/x123/west/-,
        Qwest/C/west/-} {
        \begin{scope}
          \clip ([yshift=1mm]\nodesource.north \anchor)
          rectangle ([yshift=-1mm,xshift=\direction\nextblockdistance]
                     \nodesource.south \anchor);
          \node at ([xshift=\direction\nextblockdistance]\nodesource.\anchor) (\name) {};
        \end{scope}
      }
    \end{scope}
  \end{scope}
  \draw[thick, decoration={brace,mirror},decorate]
  ([yshift=-2mm]y1.south west) -- node[anchor=north,yshift=-2pt]{$S$} ([yshift=-2mm]y2.south east) ;
  \draw[thick, decoration={brace,mirror},decorate]
  ([yshift=-2mm]y3.south west) -- node[anchor=north,yshift=-2pt]{$B\setminus S$} ([yshift=-2mm]y3.south west -| y5.south east) ;
  \foreach \name in {Peast,Pwest,Qeast,Qwest} {
    \node[draw=none,minimum width=1pt,minimum height=1pt,text height=1pt,inner xsep=1pt]
       at (\name.center) {$\ldots$};
  }
  \begin{scope}[text depth=2pt]
  \coordinate (mapXCoordinate) at ([xshift=10mm]Qeast.center);
  \node[outer sep=1pt,anchor=center] at (Peast.center -| mapXCoordinate) (X) {$C$};
  \node[outer sep=1pt,anchor=center] at (Qeast.center -| mapXCoordinate) (PY) {$\Potf C$};
  \end{scope}
  \draw[commutative diagrams/.cd, every arrow, every label] (X) edge node {$c$} (PY);
  \begin{scope}[
    bend angle=10,
    space/.style={
      draw=white,
      line width=4pt,
    },
    edge/.style={
      draw=black,
      -{>[length=2mm,width=2mm]},
      preaction={draw,-,line width=1mm,white},
      every node/.append style={
        fill=white,
        shape=rectangle,
        inner sep=1pt,
        anchor=base,
        pos=0.55,
      },
    },
    ]
    \draw[edge,bend right] (x1) to (y1);
    \draw[edge,bend left] (x1) to (y2);
    \draw[edge,bend right] (x2) to (y2);
    \draw[edge,bend left] (x2) to (y3);
    \draw[edge,bend left] (x2) to (y4);
    \draw[edge,bend right] (x3) to (y3);
    \draw[edge] (x3) to (y4);
    \draw[edge,bend left] (x3) to (y5);
  \end{scope}
  \begin{scope}[scissor lines
    ]
    \foreach \leftnode/\rightnode/\northlen/\southlen in
    {x1/x2/4mm/4mm,x2/x3/4mm/4mm,y2/y3/4mm/4mm} {
      \coordinate (northend) at ($ (\leftnode) !.5! (\rightnode) + (0,\northlen)$);
      \coordinate (southend) at ($ (\leftnode) !.5! (\rightnode) - (0,\southlen)$);
      \node[anchor=east,rotate=-90,scissors] at ([xshift=-0.25pt]northend) {\scissors};
      \draw[gapstyle] (northend) -- (southend);
      \draw[linestyle] (northend) -- (southend);
    }
  \end{scope}
\end{tikzpicture}
  \caption{Splitting of a block $B$ into smaller blocks results in further
    refinement of the block $\{x,y,z\}$}
  \label{fig:parttree}
\end{figure}
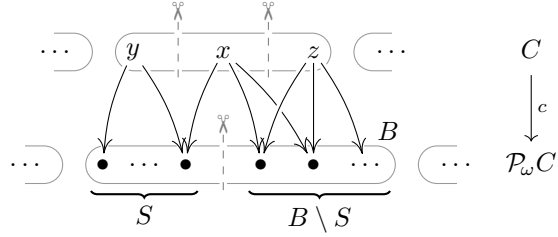

In this process of propagating inequivalences, suppose that the partition
refinement has already computed a block of states $B\subseteq C$ in its
partition and that states in $S\subseteq B$ have different behaviour from those
in $B\setminus S$ (as illustrated in \autoref{fig:parttree}). From this
information, the algorithm infers whether states $x,y\in C$ that are in the same
block and have successors in $B$ exhibit different behaviour and thus have to be
separated. Let us explain on two concrete instances how this inference is
achieved.
\begin{example}\label{E:split}
  \begin{enumerate}
  \item\label{E:split:1} We mentioned in \autoref{exManyFunctors}\ref{exManyFunctors:Pow} that
    for $F=\Powf$, a coalgebra is a finitely branching transition system. A
    partition on the state space represents a bisimulation\twnote{} if
    it has the following property: 

    \medskip\noindent\hspace{1cm}
    \begin{minipage}{15cm}
      For every pair of states $x,y$ in the same block and every block $B$:\\
      $x$ has a successor in $B$ iff $y$ has a successor in $B$.
    \end{minipage}

    \medskip\noindent
    This means that when we split the block $B$ into the two
    blocks $S$ and $B\setminus S$, then $x$ and $y$ can stay in the same block
    provided that (a)~$x$ has a successor in $S$ iff $y$ has one and (b)~$x$ has
    a successor in $B\setminus S$ iff $y$ has one. Equivalently, we can express
    these conditions by the equality
    \begin{equation}
      \label{PowChiBS}
      \Powf\fpair{\chi_S,\chi_{B\setminus S}}(c(x)) =
      \Powf\fpair{\chi_S,\chi_{B\setminus S}}(c(y)),
    \end{equation}
    where $\chi_S, \chi_{B \setminus S}\colon C\to 2$ are the usual characteristic
    functions of the subsets $S, B\setminus S\subseteq C$, respectively. Indeed, the function
    \[
      C \xrightarrow{c} \Powf C
      \xrightarrow{\Powf\fpair{\chi_S,\chi_{B\setminus S}}} \Powf (2\times 2)
    \]
    maps a state $x\in C$ to the set
    \[
      P \coloneqq \Powf\fpair{\chi_S,\chi_{B\setminus S}}(c(x)) \subseteq 2
      \times 2
    \]
    encoding whether $x$ has successors in $S$ resp.\ in $B\setminus S$. In
    fact, we see that $x$ has a successor in $S$ iff $(1,0) \in P$ and $x$ has a
    successor in $B\setminus S$ iff $(0,1) \in P$. Moreover, $x$ has a successor
    in $C \setminus B$ iff $(0,0) \in P$.  Since $S$ and $B\setminus S$ are
    disjoint, we have $(1,1)\not\in P$. Similar observations apply to~$y$. Thus,
    in order to maintain the desired property, we need to separate $x$ and $y$
    iff~\eqref{PowChiBS} holds.\footnote{Note that since the above property
      holds before splitting $B$, we have that $(0,0)$ is either contained in
      both sides of~\eqref{PowChiBS} or in neither of them.}
  \item In the example of Markov chains, i.e.~$F= \Dist$, we can make a similar
    observation. Here, $x,y\in C$ can stay in the same block if the weights of
    all transitions from $x$ to states in $S$ sum up to the same value as the
    weights of all transitions from $y$ to states in $S$ and similarly for
    $B\setminus S$; that is if
    \[
      \sum_{x'\in S} c(x)(x') = \sum_{y' \in S} c(y)(y')
      \qquad\text{and}\qquad
      \sum_{x'\in B\setminus S} c(x)(x') =\sum_{y'\in
        B\setminus S} c(y)(y').
    \]
    Analogously as in item~\ref{E:split:1}, we can equivalently express this by
    stating that $x$ and $y$ can stay in the same block iff the following
    equation holds
    \begin{equation}
      \Dist\fpair{\chi_S,\chi_{B\setminus S}}(c(x))
      = \Dist\fpair{\chi_S,\chi_{B\setminus S}}(c(y)).
      \label{DistChiBS}
    \end{equation}
    The map
    \[
      C \xrightarrow{c} FC
      \xrightarrow{\Dist\fpair{\chi_S,\chi_{B\setminus S}}} \Dist(2\times
      2)\subseteq \R^{(2\times 2)}
    \]
    sends every state $x\in C$ to the function
    $t = \Dist\fpair{\chi_S,\chi_{B\setminus S}}(c(x))\colon 2 \times 2 \to \R$ where
    \begin{itemize}
    \item $t(1,0)$ is the accumulated weight of all transitions from $x$ to
      states in $S$, 
    \item $t(0,1)$ is the accumulated weight of all transitions from $x$ to $B\setminus S$,
      
    \item $t(0,0)$ is the accumulated weight of all transitions from $x$ to
      states outside of $B$,
      
    \item $t(1,1) = 0$, since there is no element in $S \cap (B\setminus S$).
    \end{itemize}
    Two states $x,y$ that are in the same block before splitting $B$ into $S$
    and $B\setminus S$ have the same accumulated weight of transitions to $B$
    and also to $C\setminus B$, so $x$ and $y$ can stay in the same block
    iff~\eqref{DistChiBS} holds.
  \end{enumerate}
\end{example}
\noindent It is now immediate how~\eqref{PowChiBS} and~\eqref{DistChiBS} are
generalized to an arbitrary functor $F$:  States $x$ and $y$ stay in the same
block in a refinement step iff
\begin{equation}
  F\fpair{\chi_S,\chi_{B\setminus S}}(c(x))
  = F\fpair{\chi_S,\chi_{B\setminus S}}(c(y)).
  \label{FChiBS}
\end{equation}
As we have seen in~\autoref{E:split}, $(1,1)$ is never in the image of
$\fpair{\chi_S,\chi_{B\setminus S}}\colon C\to 2\times 2$, because $S$ and $B
\setminus S$ are disjoint. Hence we can restrict
its codomain to $3 = \{0,1,2\}$ by defining the map
\begin{equation}
  \chi_S^B\colon C\to 3
  \quad
  \text{for }S\subseteq B
  \quad
  \text{by }
  \qquad
  \chi_S^B(x) = \begin{cases}
    2 &\text{if }x\in S,\\
    1 & \text{if }x \in B\setminus S,\\
    0 & \text{if }x \in C\setminus B.
  \end{cases}
\end{equation}
Hence, the criterion in \eqref{FChiBS} can be simplified:
the states $x$ and $y$ stay in the same block in the refinement step iff
\begin{equation}
F\chi_S^B(c(x)) = F\chi_S^B(c(y)).\label{FChi3}
\end{equation}
We conclude that the generic partition refinement algorithm needs to compute
the value $F\chi_S^B(c(x)) \in F3$ for every state $x$. Whenever states $x$ and
$y$ are sent to different values by $F\chi_S^B\cdot c$, we know that they are
behaviourally inequivalent and need to be moved to separate blocks.

\subsection{Refinement Interfaces}
\label{sec:refintface}
Computing the values $F\chi_S^B(c(x))$ for states $x$ of interest is the
task of the refinement interface. We start with its formal definition and then
provide an informal explanation of its ingredients. 
\removebrackets
\begin{definition}[\cite{concurSpecialIssue}]\label{D:refinement-interface}
  Given an encoding $(A,\flat)$ of the set functor $F$,
  a \emph{refinement interface} for $F$ consists of 
  a set $W$ of \emph{weights} and functions
  \[
    \op{init}\colon F1×\Bag A\to W \qquad\text{and}\qquad
    \op{update}\colon \Bag A × W \to W× F3× W
  \]
  satisfying the coherence condition that there exists a family of \emph{weight
    maps} $w\colon \mathcal{P} X \to (FX \to W)$ (not themselves part of the
  interface to be implemented!), one for each set~$X$, such that
  \begin{align*}
    \op{init}\big(F!(t), \mbraces{a \mid (a,x) \in\flat (t)} \big) & = w(X)(t)\\
    \op{update}\big(\mbraces{a\mid (a,x)\in \flat(t), x\in S}, w(B)(t)\big) & =
    (w(S)(t),F\chi_S^B(t),w(B\!\setminus\! S)(t))
  \end{align*}
  for $t\in FX$ and $S\subseteq B\subseteq X$.  Here, the notation
  $\mbraces{a \mid \cdots}$ in the arguments of $\op{init}$ and $\op{update}$
  indicates \emph{multiset} comprehension, i.e.\ multiple occurrences of a label
  $a\in A$ in $\flat(t)$ result in multiple occurrences of $a$ in the bag
  $\mbraces{a \mid (a,x)\in \flat(t),\cdots}$.
\end{definition}
\noindent The intuition behind the refinement interface can be understood best
if we consider an $F$-coalgebra $c\colon C\to FC$, put $X := C$, fix a state
$x\in C$, and instantiate $t := c(x)$. As one can see from the types of
$\op{init}$ and $\op{update}$, the refinement interface is designed in such a
way that it computes values of the functor specific type~$W$ of weights that the
calling algorithm saves for subsequent calls to $\op{update}$. For every block
$B \subseteq C$, the value $w(B)(c(x)) \in W$ is the accumulated weight of edges
from $x$ to (states in) $B$ in the coalgebra $(C,c)$. In principle, values
in~$W$ can contain whatever information about the set of edges from $x$ to $B$
helps the implementation of $\op{update}$ to compute the result value of type
$F3$, which is what the caller is actually interested in. However, while more
information contained in the second argument of $\op{update}$ helps this
function to achieve this task more efficiently, both $\op{init}$ and
$\op{update}$ also have to compute values of~$W$, which may require more effort
if these values carry too much information.

This trade-off is guided by the two equational axioms for $\op{init}$ and
$\op{update}$, which represent a contract that their implementation for a
particular functor has to fulfil. The first axiom assumes that $\op{init}$
receives in its first argument the output behaviour of $x$ -- e.g.~whether $x$
is final or non-final (in the case of automata), or for $F=\Powf$ whether $x$
has any successors or is a deadlock state -- and in its second argument the bag
of labels of all outgoing edges of $x$ in the graph representation of
$(C,c)$. The axiom then requires $\op{init}$ to return the accumulated weight
$w(C)(c(x)) \in W$ of all edges from $x$ to the whole state set $C$, which
is the only block in the initial partition of $C$. This corresponds to the use
of $\op{init}$ in the actual algorithm, namely to initialize the
weight value in~$W$ that is later passed to $\op{update}$.

\begin{figure}
\centering%
  \begin{tikzpicture}
  \node (x) at (0.65,1.5) {$x$};
  \begin{scope}[
    every node/.append style={
      shape=circle,
      draw=none,
      fill=black,
      inner sep=0pt,
      outer sep=2pt,
      minimum width=1pt,
      minimum height=4pt,
      anchor=center,
    }
    ]
    \node (y1) at (0,0) {};
    \node (y2) at (1.3,0) {};
    \node (y3) at (2,0) {};
  \end{scope}
  \node[anchor=west] (y4) at (y3.east) {$\ldots$};
  \draw (y1) edge[draw=none] node {$\ldots$} (y2);
  \begin{scope}[
    every node/.append style={
      draw=black!30,
      rounded corners=2mm,
    }
    ]
    \node[fit=(y1) (y2),draw=none,inner sep=0pt] (S) {};
    \node[fit=(y3) (y4),draw=none,inner sep=0pt] (CminusS) {};
    \node[fit=(S) (CminusS)] (C) {};
  \end{scope}
  \begin{scope}[scissor lines]
    \coordinate (northend) at ($ (S.north east) !.5! (CminusS.north west) + (0,2mm)$);
    \coordinate (southend) at ($ (S.south east) !.5! (CminusS.south west) - (0,3mm)$);
    \draw[gapstyle] (northend) -- (southend);
    \draw[linestyle] (northend) -- (southend);
    \node[anchor=east,rotate=90,scissors]
    at ([xshift=0.2pt]southend) 
    {~\scissors};
  \end{scope}
  \begin{scope}[
    bend angle=10,
    space/.style={
      draw=white,
      line width=4pt,
    },
    edge/.style={
      draw=black,
      ->,
      every node/.append style={
        fill=white,
        shape=rectangle,
        inner sep=1pt,
        anchor=base,
        pos=0.55,
      },
    },
    ]
    \draw[space,bend right] (x) edge (y1);
    \draw[edge,bend right] (x) edge node[alias=s1] {$a_1$} (y1);
    \draw[space,bend left] (x) edge (y2);
    \draw[edge,bend left] (x) edge node[alias=sk] {$a_k$} (y2);
    \draw[space,bend left] (x) edge (y3);
    \draw[edge,bend left=20] (x) edge node[alias=c] {$b$} (y3);
    \coordinate (dotends) at ([xshift=8mm]c);
    \draw[edge,-,bend left] (x) edge (dotends);
    \draw[edge,-,shorten >=3mm,dotted] (dotends) edge +(4mm,-4mm);
  \end{scope}
  \draw (s1) edge[draw=none] node {$\cdots$} (sk);
  \draw (c) edge[draw=none] node {$\cdots$} (dotends);
  \begin{scope}[every node/.append style={anchor=north,yshift=-1mm,font=\footnotesize}]
    \draw[decorate,decoration={brace,mirror,amplitude=4pt}]
    ([yshift=-6mm]C.south west) -- node {$B$}
    ([yshift=-6mm]C.south east);
    \draw[decorate,decoration={brace,mirror,amplitude=4pt}]
    ([yshift=-2mm]S.south west) -- node {$S$}
    ([yshift=-2mm]S.south east);
    \draw[decorate,decoration={brace,mirror,amplitude=4pt}]
    ([yshift=-2mm]CminusS.south west) -- node {$B\setminus S$}
    ([yshift=-2mm]CminusS.south east);
  \end{scope}
\end{tikzpicture}
\caption{Computation of the value of type $F3$ for $x$}
\noshowkeys
\label{fig:splitblock}
\end{figure}
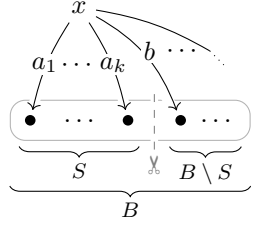
The operation \op{update} is called whenever the algorithm derives that a block
$B$ in the partition of $C$ contains behaviourally inequivalent states,
i.e.~when the block $B$ has been split into smaller blocks, including, say,
$S\subseteq B$, like in~\autoref{fig:splitblock}. This means that every state
$x'\in S$ is behaviourally inequivalent to every state $x''\in B\setminus
S$. The first parameter of $\op{update}$ is then the bag of labels of all edges
from $x$ to $S$, and the second parameter is the weight
$w(B)(c(x)) \in W$ of all edges from $x$ to $B$, which the caller has saved from
return values of previous calls to $\op{init}$ and $\op{update}$,
respectively. From only this information (in particular, $\op{update}$ does not
know $x$, $S$ or $B$ explicitly), $\op{update}$ computes the triple consisting of the
weight $w(S)(c(x))$ of edges from $x$ to $S$, the result of
$F\chi_S^B \cdot c(x)$, and the weight $w(B\setminus S)(c(x))$ of edges from $x$
to $B\setminus S$. The two weights are stored by the caller in order
to supply them to $\op{update}$ in the next refinement step, and
$F\chi_S^B(c(x))$ is used to split the block containing $x$ according
to~\eqref{FChi3}.

For a given functor $F$, it is usually easy to derive the operations $\op{init}$
and $\op{update}$ once appropriate choices of the set $W$ of weights and
the weight maps~$w$ are made. We now recall refinement interfaces for some
functors of interest; see~\cite{concurSpecialIssue} for the verification of the
axioms. 
\removebrackets
\begin{example}[\cite{concurSpecialIssue}]\label{E:refint}
  \begin{enumerate}
  \item\label{E:refint:1} For $F=\Powf$, we put $W= 2 \times \N$. For further use
    in the definition of the weight maps and the refinement interface routines,
    we define an auxiliary function
    \[
      (-) \geZero 0\colon \N \to 2\qquad\text{by }(n\geZero 0) =
      \min(n,1) =
      \begin{cases}
        1 &  \text{if $n > 0$} \\
        0 &  \text{otherwise.}
      \end{cases}
    \]
    The weight maps are defined by
    \[
      w(B)(t) = (|t\setminus B| \geZero 0,|t\cap B|).
    \]
    This records whether there is an edge to $X\setminus B$ and counts the
    numbers of edges to states in the block $B$. This number is crucial to be
    able to implement the $\op{update}$ routine which needs to return
    $\Powf\chi_S^B(c(x)) \in \Powf 3$ for a coalgebra $c\colon C\to FC$ and a
    state $x \in C$. Hence, $\op{update}$ needs to determine whether $x$ has an
    edge to $B\setminus S$ -- i.e.~whether $1\in \Powf\chi_S^B(c(x))$ -- given
    only the number~$k$ of edges from $x$ to $S$ and the weight
    $w(B)(c(x))$ (cf.~\autoref{fig:splitblock}). This task can only be
    accomplished if $w(B)(c(x))$ holds the number $n$ of edges from $x$ to $B$:
    with this information, there are edges to $B\setminus S$ iff $n - k >
    0$. Recall from \autoref{E:enc}\ref{E:enc:1} that the set of labels is $A =
    1$. Hence, every bag of labels is just a natural number because $\Bag
    A = \Bag 1 \cong \N$. Consequently, the interface routines
    \[
      \op{init}\colon \Potf 1\times \N \to 2\times \N
      \qquad\text{and}\qquad
      \op{update}\colon \N \times (2\times \N) \to (2\times \N)\times \Potf 3
      \times (2\times \N) 
    \]
    are implemented as follows:
    \[
      \op{init}(z,n) = (0,n) \qquad\text{and}\qquad
      \begin{array}[t]{r@{\,}l@{}l}
        \op{update}(n_S,(r,n_C)) &=
        \big(&(r\vee (n_{C\setminus S}\geZero 0), n_S), \\[5pt]
        && (r, n_{C\setminus S}\geZero 0, n_S\geZero 0), \\[5pt]
        && (r\vee (n_{S}\geZero 0), n_{C\setminus S})\big),
      \end{array}
    \]
    where $n_{C\setminus S} := \max(n_C - n_S, 0)$,
    $\vee\colon 2\times 2\to 2$ is disjunction, and the middle return
    value in $\Potf 3$ is written as a bit vector of length three. The axioms in
    \autoref{D:refinement-interface} ensure that
    $n_S, n_C, n_{C\setminus S}$ can be understood as the numbers of
    edges to $S$, $C$, and $C\setminus S$, respectively.
    
    The technique of remembering the number of edges from every state to
    every block is already crucial in the classical algorithm by Paige and
    Tarjan~\cite{PaigeTarjan87}. In \autoref{sec:noncan}, we will generalize
    this trick from $\Powf$ to arbitrary monoid-valued functors.

  \item\label{E:refint:2} For the group-valued functor $G^{(-)}$ for the abelian
    group $G$, we put $W=G^{(2)} \cong G \times G$, and the weight map is
    defined by
    \[
      w(B) = G^{(\chi_B)}\colon G^{(X)}\to G^{(2)}
      \qquad\text{for every subset $B\subseteq X$}.
    \]
    The refinement interface routines are implemented as follows:
    \[
    \begin{array}{l@{\quad}l}
      \op{init}\colon G^{(1)}\times \Bagf(G)
      \to G^{(2)}
      &
      \op{update} \colon \Bagf(G)\times G^{(2)}
      \to G^{(2)}
        \times G^{(3)}
        \times G^{(2)}
        \\
        \op{init}\underbrace{(g,~}_{\mathclap{G^{(!)}(c(x))}} \ell) =
        \underbrace{(0, g)}_{\mathclap{G^{(\chi_C)}(c(x))}}
        &
        \op{update}(\ell,\underbrace{(r,b)}_{\mathclap{G^{(\chi_{B})}(c(x))}}) = \big(
        \underbrace{(r + (b - \csum\ell),\csum\ell)}_{G^{(\chi_S)}(c(x))},~
        \underbrace{(r,b-\csum\ell,\csum\ell)}_{G^{(\chi^B_S)}(c(x))},~
        \underbrace{(r+\csum\ell,b-\csum\ell)}_{G^{(\chi_{B\setminus S})}(c(x))}\big),
      \end{array}
    \]
    where $\csum\colon \Bagf G\to G$ is the summation map of \autoref{N:csum}. 
    The terms under the braces only serve as the intuition when
    considering a coalgebra $c\colon C\to FC$ and a subblock $S\subseteq B$ of a
    block $B\subseteq C$. The function $\op{init}$ is called with the bag $\ell$
    of labels of outgoing transitions of some element $x\in C$ and the sum $g$
    of all these labels. Since $w(C)(c(x)) = (0, G^{(!)}(c(x)))$ for the whole
    set $C$, the $\op{init}$ function simply returns $(0,g)$.

    In the above $\op{update}$ routine we have used that $G$ has inverses.  In
    \autoref{sec:monoidvalued}, we will define a refinement interface for the
    monoid-valued functor $M^{(-)}$ for monoids $M$ that are not groups, using
    additional data structures to make up for the lack of inverses in $M$.

  \item As special instances of the previous item we obtain refinement
    interfaces for the functors $\R^{(-)}$ and $\Z^{(-)}$.

  \item For a polynomial functor $F = \Sigma$, we put $W = \Sigma 2$ and 
    \[
      w(B) = \Sigma\chi_B\colon \Sigma X\to \Sigma 2
      \qquad\text{for every subset $B\subseteq X$}.
    \]
    For a coalgebra $c\colon C\to \Sigma C$ and $B\subseteq C$, this means that
    $w(B)(c(x))$ consists of an operation symbol $\sigma\in \Sigma$ and a bit
    vector of length $\op{ar}(\sigma)$. The bit vector specifies which successor
    states of $x$ are in the set $B$. Recall from \autoref{E:enc}\ref{E:enc:3}
    that the encoding of $\Sigma$ uses as labels $A=\N$. With $1 =\{*\}$, the
    $\op{init}$ routine is given by
    \[
    \begin{array}{l@{\quad}l}
      \op{init}\colon \Sigma 1\times \Bagf \N \to \Sigma 2
      \\
      \op{init}(\sigma(*,\ldots,*), f) = \sigma(1,\ldots,1)
    \end{array}
    \]
    For $\op{update}$, we first define the map that computes the middle
    component of the result:
    \begin{align*}
        \op{update}'&\colon \Bagf\N \times \Sigma 2 \to \Sigma 3
      \\
      \op{update}'&(I, \sigma(b_1,\ldots,b_n))
      = \sigma(b_1 + (1\in I),\ldots,b_i + (i\in I),\ldots,b_n + (n\in I))
    \end{align*}
    Here $b_i + (i\in I)$ means $b_i + 1$ if $i\in I$ and $b_i$ otherwise. For
    subsets $S\subseteq B\subseteq X$ and a state $y\in X$, this sum computes
    $\chi_S^B(y)$ from $\chi_S(y)$ (given by $i\in I$) and $\chi_B(y)$ (given by
    the bit $b_i$). From the value in $\Sigma 3$ thus computed, we can derive
    the other components of the result of $\op{update}$. For $k\in 3$, let
    \((k=)\colon 3\to 2 \) be the map that compares its parameter with $k$, i.e.
    \[
      (k=)(k') =\begin{cases}
        1 &\text{if }k = k' \\
        0 &\text{otherwise.}
      \end{cases}
    \]
    The $\op{update}$ routine now calls $\op{update}'$ and derives the values of type
    $\Sigma 2$:
    \begin{align*}
      &\op{update} =  \big(
      \Bagf\N \times \Sigma 2 \xrightarrow{\op{update}'}
      \Sigma 3
      \xrightarrow{\fpair{\Sigma (2=),\id_{\Sigma 3}, \Sigma (1=)}}
      \Sigma 2\times \Sigma 3
        \times \Sigma 2\big)
          \\
        &\op{update}(I, t) = \fpair{\Sigma (2=),\id_{\Sigma
        3}, \Sigma (1=)} (\op{update}'(I, t))
    \end{align*}

    \takeout{}%
    \end{enumerate}
\end{example}\medskip

\noindent In order to ensure that iteratively splitting blocks using $F\chi_S^B$
in each iteration correctly computes the minimization of the given coalgebra, we
need another property of the functor $F$:

\removebrackets
\begin{definition}[{\cite[Def.~5.1]{concurSpecialIssue}}]\label{D:zippable}
  A functor~$F$ is \emph{zippable} if for all sets $X$ and $Y$ the map
  \[
    {\fpair{F(X+\,!\,),F(\,!+Y)}}\colon {F(X+Y)}\longrightarrow F(X+1)\times
    F(1+Y)
  \]
  is injective.
\end{definition}
All functors mentioned in \autoref{exManyFunctors} are zippable. Moreover,
zippable functors are closed under products, coproducts (both formed
point-wise), and subfunctors~\cite[Lemma~5.4]{concurSpecialIssue}. However, they
are not closed under functor composition: for example, $\Powf\Powf$ fails to be
zippable~\cite[Example~5.9]{concurSpecialIssue}. We deal with this problem by a
reduction discussed in \autoref{sec:des} below.

For zippable set functors $F$ with a refinement interface, we have presented a
partition refinement
algorithm~\cite[Algorithm~7.7]{concurSpecialIssue}.\smnote{}  \takeout{}%
The main correctness result states that for a zippable functor equipped with a
refinement interface, our algorithm correctly minimizes the given coalgebra. The
low time complexity of our algorithm hinges on the time complexity of the
implementations of $\op{init}$ and $\op{update}$.  We have shown
previously~\cite[Theorem~7.16]{concurSpecialIssue} that if both $\op{init}$ and
$\op{update}$ run in linear time in its input of type $\Bag A$ \emph{alone}
(i.e.~independently of the input coalgebra size), then our generic partition
refinement algorithm runs in time $\CO((m + n)\cdot \log n)$ on coalgebras with
$n$ states and $m$ edges (cf.~\autoref{D:enc}). In order to cover instances
where the run time of $\op{init}$ and $\op{update}$ depends also on the input
coalgebra, we make this dependence formally explicit:
\begin{definition}\label{D:runTimeFactor} The refinement interface
  for a functor $F$ \emph{has run time factor} $\rifactor(c)$ if for every map
  $c\colon X\to FY$ (in particular for every coalgebra $c\colon C\to FC$),
  \begin{enumerate}
  \item the following calls to $\op{init}$ and $\op{update}$ run in
    time $\CO(|\ell|\cdot \rifactor(c))$ for $x\in X$, $t=c(x)$, and
    $S\subseteq B\subseteq X$:
  \[
    \begin{array}[b]{@{\,}ll}
    \op{init}\big(F!(t), \ell \big)
      &\text{where }\ell = \mbraces{a \mid (a,x) \in\flat (t)}
    \\[1mm]
    \op{update}\big(\ell, w(B)(t)\big)
      &\text{where }\ell = \mbraces{a\mid (a,x)\in \flat(t), x\in S}.
    \end{array}
  \]
    \takeout{}%

\item\label{D:runTimeFactor:2} equality of values in\twnote{}
  $\{F\chi_S^B(c(q))\mid q\in X, S\subseteq B\subseteq X\} \subseteq F3$ can be
  checked in time $\CO(\rifactor(c))$.
  \end{enumerate}
  If $\rifactor(c)$ only depends on the number of states $n$ and number of
  transitions $m$ in $c$, then we write $\rifactor(n,m)$ in lieu of
  $\rifactor(c)$. Note that the above calls to $\op{init}$ and $\op{update}$
  are precisely those from the axioms of the refinement interface in
  \autoref{D:refinement-interface}.
\end{definition}
\begin{example}\label{expl:run-time-factor}
  The powerset functor $\Powf$ and the group-valued functors $G^{(-)}$ have run
  time factor ${p(c) = 1}$ \cite[Examples~6.11]{concurSpecialIssue}. For a
  signature $\Sigma$ where the arity of operation symbols is bounded by
  $b \in \N$, the refinement interface has run time factor
  $\rifactor(c) = 1$. Otherwise we define the \emph{rank} of a finite
  $\Sigma$-coalgebra $c\colon C\to \Sigma C$ to be the maximal arity that
  appears in $c$:
  \[
    \op{rank}(c) = \max\{\op{ar}(\sigma) \mid \text{$\sigma(x_1, \ldots, x_n) = c(x)$ for
      some $x_1, \ldots, x_n$ and $x$ in $C$}\}.
  \]
  It is easy to see that the refinement interface for $\Sigma$
  (see~\cite[Examples~6.11.3]{concurSpecialIssue}) has run time factor $p(c)
  = \op{rank}(c)$.
\end{example}
Since we can now describe the run time of the refinement-interface in a more
fine-grained way, we can lift this to the run time analysis of the overall
partition refinement algorithm.
\begin{theorem}\label{T:comp}
  Let $F$ be a zippable functor equipped with a refinement interface with 
  run time factor $\rifactor(c)$. Then the algorithm computes the behavioural
  equivalence relation on an input $F$-coalgebra $c\colon C\to FC$ with $n$
  states and $m$ transitions in $\CO((m+n)\cdot \log n \cdot \rifactor(c))$
  steps.
\end{theorem}
\begin{proof}
  The case where $\rifactor(c) = 1$ is proved in
  \cite[Theorem~6.22]{concurSpecialIssue}. We reduce the general case to this
  one as follows. Observe that the previous complexity analysis counts the
  number of basic operation performed by the algorithm (e.g.~comparisons of
  values of type $F3$) including those performed by $\op{init}$ and
  $\op{update}$. In that analysis $\op{init}$ and $\op{update}$ were assumed to
  have run time $\CO(|\ell|)$, and the total number of basic operations of the
  algorithm is then $\CO((m+n) \cdot \log n)$.

For the reduction, we consider `macro' operations that run in
$\CO(\rifactor(c))$ time. In particular, every constant-time operation that is
performed in the algorithm can be viewed as a macro
performing precisely this single operation. Then we can view the generalized run time
assumptions on the refinement-interface as follows:
\begin{enumerate}
\item all calls to $\op{init}$ and $\op{update}$ on $\ell\in \Bagf A$ perform
  $\CO(|\ell|)$ macro operations (each of which takes $\CO(\rifactor(c))$ time).
\item all values of type $F3$ that arise during the execution of the algorithm
  are in the set in \autoref{D:runTimeFactor}\ref{D:runTimeFactor:2}. Hence,
  every comparison of such values is done in one macro operation, which takes
  $\CO(\rifactor(c))$ steps.
\end{enumerate}
By the previous complexity analysis for $\rifactor(c) =1$, the
partition refinement for a coalgebra with $n$ states and $m$ edges performs
{$\CO((m+n)\cdot \log n)$} macro calls. Thus, the overall run time lies in
\[
  \CO((m+n) \cdot \log n \cdot \rifactor(c))
\]
as desired.
\end{proof}
Obviously, for $\rifactor(c)\in\CO(1)$, we obtain the
previous complexity.
\begin{example} For a (possibly infinite) signature $\Sigma$, the coalgebraic
  partition refinement runs, by \autoref{T:comp} and
  \autoref{expl:run-time-factor}, in time $\CO(r\cdot (m+n)\cdot \log n)$ for an
  input coalgebra $c\colon C\to \Sigma C$ with rank~$r$,~$n$ states, and~$m$
  edges. Note that every state $x\in C$ has at most $r$ many outgoing edges in
  the graph representation. Hence, we have $m\le r\cdot n$ so that the time
  complexity may be simplified to $\CO(r^2\cdot n\cdot \log n)$.
\end{example}

\noindent In \autoref{sec:monoidvalued} we will discuss how \autoref{T:comp}
instantiates to the example that mainly motivated the above generalization of
the complexity analysis: weighted systems with weights from an unrestricted
commutative monoid.

\subsection{Combining Refinement Interfaces}\label{sec:des}%
In addition to supporting genericity via direct implementation of the
refinement interface for basic functors, our tool is \emph{modular} in
the sense that it automatically derives a refinement interface for
functors built from the basic ones according to the grammar
\eqref{termGrammar}. In other words, for such a combined functor the user does
not need to write a single line of new code. Moreover, when the user implements
a refinement interface for a new basic functor, this automatically extends the
effective grammar. For example, our tool can minimize systems of type
\[
FX= \Dist(\N \times \Powf X\times \Bag X).
\]
However, while all basic functors from which $F$ is formed are zippable (see
\autoref{D:zippable}), there is no guarantee that $F$ is so because zippable
functors are not closed under functor composition in general. In order to
circumvent this problem, a given $F$-coalgebra is transformed into one for the
functor
\[
  F' X = \Dist X +  (\N \times X \times X)  + \Powf X + \Bag X.
\]
\begin{figure}
  \centering
  \begin{tikzpicture}
    \node[unary] (Potf) {\Dist};
    \node[binary] (times) at (Potf.in) {\Sigma};
    \node[unary] (Bagf) at ([yshift=2mm]times.in1) {\Powf};
    \node[unary] (Dist) at ([yshift=-2mm]times.in2) {\Bag};
    \path[draw=black!50]
      (times.out) edge node[above] {$Y$} (Potf.in)
      (Bagf.out) edge node[above] {$Z_1$} (times.in1)
      (Dist.out) edge node[below] {$Z_2$} (times.in2)
      (Dist.in) edge node[above] {$X$} +(8mm,0)
      (Potf.out) edge node[above] {$X$} +(-8mm,0)
      (Bagf.in) edge node[above] {$X$} +(8mm,0);
  \end{tikzpicture}
  \caption{Visualization of $FX = \Dist (\Sigma(\Powf X, \Bag X))$ for $\Sigma (Z_1,Z_2) = \N\times Z_1\times Z_1$}
  \label{fig:composedFunctor}
\end{figure}%
This functor is obtained as the sum of all basic functors involved in $F$,
i.e.~of all the nodes in the visualization of the functor term $F$
(\autoref{fig:composedFunctor}). Then the components of the refinement
interfaces of the four functors involved, viz.\ $\Dist$, $\Sigma$, $\Powf$, and
$\Bag$, are combined by disjoint union $+$. The transformation of a finite
coalgebra $c\colon C \to FC$ into a finite $F'$-coalgebra introduces a set of
intermediate states for each edge in the visualization of the term $F$; we have
labelled the edges in \autoref{fig:composedFunctor} by these sets. The
construction starts with $X := C$ and constructs a finite $F'$-coalgebra on the
set $C' := X + Y + Z_1 + Z_2$ as follows. The set $Y$ contains an intermediate
state for every $\Dist$-edge out of a state $x \in X$, i.e.
\[
  Y = \{ y \mid c(x)(y) \neq 0\}
  \subseteq \N\times \Powf X\times \Bagf X.
\]
This also yields a map $c_X\colon X \to \Dist Y$. Furthermore, intermediate
states in~$Y$ have successors in $\N \times \Powf X \times \Bagf X$, and by a
similar definition as for $Y$, we obtain finite sets
\[
  Z_1\subseteq \Powf X
  \qquad
  Z_2\subseteq \Bagf X
\]
and a map $c_Y\colon Y \to \N \times Z_1 \times Z_2$. Finally, intermediate
states in $Z_1$ and $Z_2$ have successors in $\Powf X$ and $\Bag X$,
respectively, which yields (inclusion) maps $c_{Z_1}\colon Z_1 \to \Powf
X$ and $c_{Z_2}\colon Z_2\to \Bag X$. Putting these maps together we obtain a finite
$F'$-coalgebra
\[
  \underbrace{X + Y + Z_1 + Z_2}_{C'}
  \xrightarrow{c_X + c_Y + c_{Z_1} + c_{Z_2}}
  \Dist Y + \N \times Z_1 \times Z_2 + \Powf X + \Bag X
  \xrightarrow{\mathsf{can}}
  \underbrace{F'(X + Y + Z_1 + Z_2)}_{FC'},
\]
where $\mathsf{can}$ is the canonical inclusion map. The minimization of this
$F'$-coalgebra  yields the minimization of the given 
$F$-coalgebra $(C,c)$. The details of the construction in full generality and its
correctness are established in~\cite[Section~8]{concurSpecialIssue}.

\begin{remark}\label{R:runTimeFactor}
  We show in the cited work that we can derive a refinement interface for $F'$ from
  the refinement interfaces of the basic functors used. It is easy to see that
  the run time factor of the refinement interface of the above $F'$ is
  \[
    p(c) = p_{\Dist}(c_X) + p_{\Sigma}(c_Y) + p_{\Powf}(c_{Z_1}) + p_{\Bagf}(c_{Z_2})
  \]
  where the summands on the right-hand side are the run time factors of the
  respective refinement interfaces of the building blocks. Note that here we use
  the full generality of \autoref{D:runTimeFactor}, i.e.~that $c_X$, $c_Y$,
  $c_{Z_1}$ and~$c_{Z_2}$ are not required to be coalgebras but only maps of the
  shape $X\to HY$ for the relevant functor~$H$.  \takeout{}
\end{remark}

\subsubsection{Combination by product}
\label{sec:combprod}

\copar moreover implements a further optimization of this procedure
that leads to fewer intermediate states in the case of polynomial
functors~$\Sigma$: Instead of putting the refinement interface
of~$\Sigma$ side by side with those of its arguments, \copar includes
a systematic procedure to combine the refinement interfaces of the
arguments of~$\Sigma$ into a single refinement interface.  For
instance, starting from $FX=\Dist(\N \times \Powf X\times \Bag X)$ as
above, a given $F$-coalgebra is transformed into a coalgebra for
the functor
\[
F'' X = \Dist X + \N \times \Powf X \times \Bag X,
\]
effectively inducing intermediate states in~$Y$ as above but avoiding~$Z_1$
and~$Z_2$. In order to run the generic partition refinement algorithm for
$F''$-coalgebras, we need a refinement interface for $F''$. \copar derives a
refinement interface for $F''$ by first combining the refinement interfaces of
$\Powf$, $\Bagf$, and that of the constant functor $X\mapsto \N$, yielding a
refinement interface for $X\mapsto \N\times \Powf X \times \Bagf X$. Then, this
refinement interface is combined with that of $\Dist$, finally yielding one for
$F''$. The combination of refinement interfaces along coproducts of functors is
already described in \cite[Sec.~8.3]{concurSpecialIssue}; in the following, we
describe how refinement interfaces are combined along cartesian
product~$\times$.

\begin{construction}\label{C:constr}
  Suppose we are given a finite family of functors
  \[
    F_i\colon \Set\to \Set, \qquad i\in I,
  \]
  such that each $F_i$ has the encoding
  $\flat_i\colon F_i X\to \Bagf(A_i\times X)$ with label set $A_i$ and is
  equipped with the refinement interface
  \[
    \op{init}_i\colon F_i1×\Bag A_i\to W_i, \qquad
    \op{update}_i\colon \Bag A_i × W_i \to W_i × F_i3 × W_i,\qquad
    w_i\colon \mathcal{P} X \to (F_iX \to W_i).
  \]
  We construct an encoding and a refinement interface for
  $FX = \prod_{i \in I} F_iX$ as follows. The encoding of $F$ is given by taking
  the disjoint union of the label sets $A_i$ and the obvious component-wise
  definition of $\flat$:
  \begin{equation}
    A=\coprod_{i\in I} A_i
    \qquad
    \flat\colon \prod_{i\in I} F_i X\longrightarrow \Bagf(A\times X)
    \qquad
    \underbrace{\flat(t)}_{\mathclap{A\times X\to \N}}(\inj_i(a),x) = \underbrace{\flat_i(\pr_i(t))}_{(A_i\times X) \to \N}(a,x)
    \qquad\text{ for every set }X.
    \label{defProdEncoding}
  \end{equation}
  The set $W$ of weights consists of tuples of weights in $W_i$, and the weight
  function simply applies the map $w_i(C)\colon F_i X\to W_i$ in the $i$-th component:
  \[
    W = \prod_{i\in I} W_i
    \qquad
    w(C)\colon~~
    \prod_{i\in I} F_i X
    \xrightarrow{~\prod_{i\in I}w_i(C)~}
    \prod_{i\in I} W_i
    \quad \text{for all sets $X$ and }C\subseteq X.
  \]
  The refinement interface routines of $F$ now have the following types: 
  \[
    \begin{array}{r@{\,}r@{\,}l}
      \op{init}\colon& (\sortprod[i]F_i1)\times \Bagf (\coprod_{i\in I}A_i)&\to
      \sortprod[i]W_i 
      \\[1mm]
      \op{update}\colon& \Bagf (\coprod_{i\in I}A_i) \times \prod_{i\in I}W_i &\to \prod_{i\in I}W_i\times \prod_{i\in I}F_i 3\times \prod_{i\in I}W_i.
    \end{array}
  \]
  For their definition, we introduce the following auxiliary function $\pi_i$,
  which restricts bags of labels to only those labels that come from $A_i$:
  \[\textstyle{
    \pi_i\colon \Bagf (\coprod_{j\in I} A_j) \to \Bagf A_i
    \qquad
    \underbrace{\pi_i(f)}_{A_i\to \N}(a) = \underbrace{f}_{A\to \N}(\inj_i(a))
    \qquad
    \text{for every }i\in I.}
  \]
  Then we define $\op{init}$ by 
  \begin{equation*}
    \op{init}
    \textstyle{
      \colon \prod_{i\in I}F_i1 \times \Bagf (\coprod_{i\in I}A_i)
      \xrightarrow{\fpair{\op{init}_i\cdot (\pr_i\times \pi_i)}_{i\in I}}
      \prod_{i\in I}W_i},
  \end{equation*}
  i.e.\
  \begin{equation*}
    \big(\op{init}((t_j)_{j\in I},\ell)\big)_{i} = \op{init}_i(t_i, \pi_i(\ell))
    ~~~ \text{ for every }i\in I,
  \end{equation*}
  and we define $\op{update}$ as the composite
  \[
    \begin{array}{l}
      \op{update} = (
      \textstyle{
      \Bagf (\coprod_{i\in I}A_i) \times \prod_{i\in I}W_i
      \xrightarrow{\fpair{\op{update}_i\cdot (\pi_i\times \pr_i)}_i}
      \prod_{i\in I}\big(W_i\times F_i 3\times W_i\big)
      \xrightarrow{\;\phi\;}
      W \times F3 \times W
      })
    \end{array}
  \]
  where $\phi$ is the obvious bijection reordering tuples in the evident
  way.  \takeout{}%
\end{construction}
\begin{proposition}\label{prodInterface}
  Let $F_i\colon \Set \to \Set$, $i \in I$, be equipped with refinement
  interfaces with run time factors $\rifactor_i(c)$.
  Then \autoref{C:constr} defines a refinement interface for $F = \sortprod F_i$
  with  run time factor
  \[
    \max\{\rifactor_i(c) \mid i \in I\}.
  \]
\end{proposition}
\noindent
In particular, if the refinement interface of every $F_i$ has run time factor
$\rifactor(c) = 1$, then so does the refinement interface for
$F=\sortprod F_i$.
\begin{proof}
  To simplify notation in the composition of maps, we define the following
  filter map for every $S\subseteq X$ and $i\in I$:
  \[
    f_S\colon \Bagf(A_i\times X) \to \Bagf A_i
    \qquad
    f_S(\,\underbrace{g}_{\mathclap{A_i\times X\to \N}}\,)
    = \big( a \mapsto \sum_{\substack{x\in S}} g(a,x)\big)
    = \mbraces{ a \mid (a,x)\in g, x\in S}.
  \]
  Using the filter maps $f_S$, we can rephrase the axioms of the
  refinement interface in \autoref{D:refinement-interface} as the commutativity
  of the following diagrams
  \begin{equation}
    \begin{tikzcd}
      FX
      \arrow{d}[swap]{\fpair{F!,f_X\cdot \flat}}
      \arrow{dr}{w(X)}
      \\
      F1\times \Bagf A
      \arrow{r}[pos=0.3]{\op{init}}
      & W
    \end{tikzcd}
    \quad\text{and}\quad
    \begin{tikzcd}
      FX
      \arrow{d}[swap]{\fpair{f_S\cdot \flat,w(B)}}
      \arrow{dr}{\fpair{w(S),F\chi_S^B,w(B\setminus S)}}
      \\
      \Bagf A\times W
      \arrow{r}{\op{update}}
      &[2mm] W \times F3 \times W
    \end{tikzcd}
    \quad
    \text{for all $S\subseteq B \subseteq X$}.
    \label{refIntfaceDiag}
  \end{equation}
  In the proof that these diagrams commute, we will use the equalities
  \begin{equation}
    \begin{tikzcd}%
      \sortprod[j] F_j X
      \arrow{r}{\pr_i}
      \arrow{d}[swap]{\flat}
      &
      F_i X
      \arrow{d}{\flat_i}
      \\
      \Bagf(\sortcoprod[j] A_j\times X)
      \arrow{d}[swap]{f_S}
      & \Bagf(A_i\times X)
      \arrow{d}{f_S}
      \\
      \Bagf (\sortcoprod[j] A_j)
      \arrow{r}{\pi_i} %
      & \Bagf A_i
    \end{tikzcd}
    \qquad
    \text{for $i\in I$ and $S\subseteq X$}.
    \label{eqPrFilter}
  \end{equation}
  This diagram commutes because for  $t\in \prod_{j\in I} F_j X$
  and $a\in A_i$, we have
  \begin{align*}
    (\pi_i\cdot f_S\cdot \flat)(t)(a)
    = \pi_i(f_S(\flat(t)))(a)
    &= f_S(\flat(t))(\inj_i(a)) \tag{def.~of~$\pi_i$}
      \\
    &= \sum_{x\in S} \flat(t)(\inj_i(a),x)
      \tag{def.~of~$f_S$}
      \\
    &= \sum_{x\in S} \flat_i(\pr_i(t))(a,x)
      \tag{def.~of $\flat$, see~\eqref{defProdEncoding}}
      \\
    &= f_S(\flat_i(\pr_i(t)))(a)
      \tag{def.~of $f_S$}
      = (f_S\cdot \flat_i\cdot \pr_i)(t)(a)
  \end{align*}
  Moreover, we will use that for every family $(Y_j)_{j\in I}$ of sets,
  every map $b\colon B\to B'$, and every $i\in I$, the diagram
  \begin{equation}
    \label{prodStrengNat}
    \begin{tikzcd}
      (\sortprod[j] Y_j) \times B
      \arrow{d}[swap]{\fpair{\pr_j\times \id_B}_{j\in I}}
      \arrow[shiftarr={yshift=8mm}]{rr}{\pr_i\times b}
      \arrow{r}{\pr_i\times \id_B}
      &[5mm] Y_i\times B
      \descto[xshift=-4mm]{dr}{naturality\\ of $\pr_i$}
      \arrow{r}{\id_{Y_i}\times b}
      & Y_i\times B'
      \\
      \sortprod[j](Y_j\times B)
      \arrow{rr}[swap]{\sortprod[j](\id_{Y_j}\times b)}
      \arrow{ur}{\pr_i}
      & & \sortprod[j] (Y_j\times B')
      \arrow{u}[swap]{\pr_i}.
    \end{tikzcd}
  \end{equation}
  commutes, as verified by straightforward evaluation of the maps involved. The
  categorically-minded reader will notice that commutation of the right-hand
  inner quadrangle is simply naturality of the projection maps~$\pr_i$, as
  indicated in the diagram; we will use this property of~$\pr_i$ again later,
  referring to it as ``naturality of $\pr_i$'' (without requiring further
  understanding of the concept of natural transformation). Furthermore, we
  clearly have commutative squares
  \begin{equation}
    \label{prodStrengNatSwapped}
    \begin{tikzcd}
      B\times (\sortprod[j] Y_j)
      \arrow{d}[swap]{\fpair{\id_B\times \pr_j}_{j\in I}}
      \arrow{r}{b\times \pr_i}
      &[10mm]
      B'\times Y_i
      \\
      \sortprod[j](B\times Y_j)
      \arrow{r}{\sortprod[j](b\times \id_{Y_j})}
      & \sortprod[j] (B'\times Y_j)
      \arrow{u}[swap]{\pr_i}
    \end{tikzcd}
    \qquad\text{for every $i \in I$}.
  \end{equation}
  We are ready to verify the axioms of the refinement interface. To this end, we
  use that the product projections $\pr_i$, $i \in I$, form a jointly injective
  family. This means that for every pair of maps $f, g\colon Z \to \prod_{j\in
    I} Y_j$ we have that
  \begin{equation}\label{eq:jinj}
    \text{$\pr_i \cdot f = \pr_i \cdot g$ for all $i \in I$ implies $f = g$}.
  \end{equation}
  \paragraph{Axiom for \op{init}.} For every $i \in I$, the outside of the
  following diagram commutes because all its inner parts commute, for the
  respective indicated reasons:
  \[
    \begin{tikzcd}[column sep=21mm,row sep=7mm]
        \sortprod[j] F_j X
        \arrow{d}[left, pos=0.4]{\begin{array}{r}
            \fpair{
            F!,\,
            f_X\cdot \flat}
            \end{array}}
          \arrow[shiftarr={yshift=8mm}]{rr}{w(X)=\sortprod[j] w_j(X)}
          \descto{dr}{\eqref{eqPrFilter}}
        \descto[yshift=4mm]{rr}{Naturality of $\pr_i$}
        & |[alias=FiX]| F_iX
        \arrow[<-]{l}[swap]{\pr_i}
        \arrow{d}[left, pos=0.4]{\begin{array}{r}
            \fpair{
            F_i!,\,
            f_X \cdot \flat_i}
            \end{array}}
        \arrow[horizontal then diagonal]{dr}{w_i(X)}
        \descto[xshift=-6mm]{dr}{Axiom\\ for $\op{init}_i$ \eqref{refIntfaceDiag}}
        & \sortprod[j] W_j
          \arrow{d}{\pr_i}
        \\
        \sortprod[j] F_j1 × \Bagf (\sortcoprod[j]A_j)
        \arrow{d}[swap]{\fpair{\pr_j\times \id}_{j\in I}}
        \arrow[to path={
          -- ([xshift=-10mm]\tikztostart.west)
          |- ([yshift=-5mm]\tikztotarget.south) \tikztonodes
          -- (\tikztotarget)
        }, rounded corners]{drr}[pos=0.8,below]{\op{init}}
        \descto{dr}{\eqref{prodStrengNat} for $Y_j = F_j1$, $b=\pi_i$}
        & F_i1 × \Bagf A_i
            \arrow{r}{\op{init}_i}
            \arrow[<-]{d}{\pr_i}
            \arrow[<-]{l}[swap]{\pr_i\times \pi_i}
            \descto{dr}{Naturality of $\pr_i$}
        & |[alias=Wi]| W_i
            \arrow[<-]{d}{\pr_i}
        \\
        \sortprod[j] \big(F_j1 × \Bagf (\sortcoprod[k]A_k)\big)
        \arrow{r}[swap]{\sortprod[j](\id\times\pi_j)}
        & \sortprod[j](F_j 1\times \Bagf A_j)
        \arrow{r}[swap]{\sortprod[j] \op{init}_j}
        & \sortprod[j] W_j
    \end{tikzcd}
  \]
  By~\eqref{eq:jinj}, we obtain commutation of the left-hand triangle
  in~\eqref{refIntfaceDiag} as desired.
  
  \paragraph{Axiom for \op{update}.}
  For all $S\subseteq B\subseteq X$, the outside of the diagram below commutes
  because all its inner parts commute for the respective indicated reasons:
  \[
    \begin{tikzcd}[column sep=19mm,row sep=10mm]
    \smash{\sortprod[j]} F_i X
    \arrow{d}[pos=0.4]{\fpair{f_S\cdot \flat,w(B)}}
    \descto{dr}{\eqref{eqPrFilter}}
    \arrow[shiftarr={yshift=8mm}]{rr}{\sortprod[j]\fpair{w_j(S),F_j\chi_S^B,w_j(B\setminus
        S)}}
    \descto[yshift=4mm]{rr}{Naturality of $\pr_i$}
    &[5mm]
    F_iX
    \arrow[
      horizontal then diagonal,
    ]{dr}[pos=0.7,above,yshift=0mm]{\fpair{w_i(S),F_i\chi_S^B, w_i(B\setminus S)}}
    \arrow{d}[pos=0.4,left]{\fpair{f_S\cdot \flat_i,w_i(B)}}
    \arrow[<-]{l}[swap]{\pr_i}
    \descto[yshift=0mm,xshift=-5mm]{dr}{Axiom for \\ $\op{update}_i$ \eqref{refIntfaceDiag}}
    &[4mm]  \sortprod[j] (W_j × F_j3 × W_j)
    \arrow{d}{\pr_i}
    \\
    \Bagf (\sortcoprod[j]A_j) × \sortprod[j] W_j
    \descto{dr}{\eqref{prodStrengNatSwapped} for $Y_j = W_j$, $b=\pi_i$}
    \arrow{d}[swap]{\fpair{\id\times \pr_j}_{j\in I}}
    \arrow[to path={
      -- ([xshift=-5mm]\tikztostart.west)
      |- ([yshift=-5mm]\tikztotarget.south) \tikztonodes
      -- (\tikztotarget)
    }, rounded corners]{drr}[pos=0.8,swap]{\phi^{-1}\cdot \op{update}}
    & \Bagf A_i × W_i
    \arrow[<-]{l}[swap,yshift=0mm]{\pi_i\times \pr_i}
    \arrow[<-]{d}{\pr_i}
    \arrow{r}[yshift=0mm]{\op{update}_i}
            \descto{dr}{Naturality of $\pr_i$}
    & W_i × F_i3 × W_i
    \arrow[<-]{d}{\pr_i}
    \\
    \smash{\sortprod[j]} (\Bagf(\sortcoprod[k]A_k) \times W_j)
    \arrow{r}[swap]{\sortprod[j] \pi_j\times \id}
    & \sortprod[j] (\Bagf A_j × W_j)
    \arrow{r}[swap]{\sortprod[j]\op{update}_j}
    & \sortprod[j] (W_j × F_j3 × W_j)
    \end{tikzcd}
  \]
  Using~\eqref{eq:jinj} and then composing with~$\phi$ on the left, we obtain
  \begin{align*}
    \op{update}\cdot \fpair{f_S\cdot \flat, w(B)}
    &= \phi\cdot \sortprod[j]\fpair{w_j(S),F_j\chi_S^B,w_j(B\setminus
        S)}
      \\
    &= \fpair{\sortprod[j]w_j(S),\sortprod[j]F_j\chi_S^B,\sortprod[j]w_j(B\setminus
        S)}
     = \fpair{w(S),F\chi_S^B,w(B\setminus S)},
   \end{align*}
   which is the desired right-hand triangle in~\eqref{refIntfaceDiag}.

   \paragraph{Run time factor.} Both $\op{init}$ and $\op{update}$ preprocess
   their parameters in linear time (via $\pi_i$ and by accessing elements of
   tuples), before calling the $\op{init}_i$ and $\op{update}_i$ routines of all
   $F_i$. Since $|I|$ is constant, this results in a run time of
   $\CO(|\ell|\cdot \max_{i\in I} \rifactor_i(c))$ for
   $\ell \in \Bagf (\sortcoprodImpl A_i)$.

   We order $F3$ lexicographically by assuming a total order on the index set
   $I$:
   \[
     x < y
     \text{ in }\sortprod F_i 3
     \quad\text{iff}\quad
     \text{ there is }i \in I\text{ with }
     \pr_i(x) < \pr_i(y)
     \text{ and }
     \pr_j(x) = \pr_j(y)
     \text{ for all }j < i.
   \]
   This comparison takes time $\CO(\max_{i\in I}\rifactor_i(c))$, again because
   $|I|$ is constant.
\end{proof}

\begin{remark}\label{R:rifactorSortMax}
  In summary we obtain that the run time factor $\rifactor(c)$ for a composite
  functor~$F$ is the maximum of the respective run time factors of the
  refinement interfaces of the basic functors from which~$F$ is
  built. Specifically, suppose that~$F$ is built from basic functors
  $G_1,\ldots,G_n$ using composition, product, and coproduct, and that
  $G_1,\ldots,G_n$ have refinement interfaces with respective run time factors
  $\rifactor_1(c),\ldots,\rifactor_n(c)$. Then the modularity mechanism
  decomposes an $F$-coalgebra $c\colon C\to FC$ into maps
  $f_i\colon X_i\to G_iY_i$, for $1\le i\le n$ (i.e.~one map per block in the
  illustration in \autoref{fig:composedFunctor}). The run time factor for the
  refinement interface arising by modular construction is given by
  \begin{equation}
    \rifactor(c) = \max_{1\le i\le n} \rifactor_i(f_i).
    \label{eq:rifactorSort}
  \end{equation}
  This is because the refinement interface for a particular functor $F_i$ only
  sees the labels and weights for the map $f_i$ and never those from other sorts.
\end{remark}

\subsection{Implementation Details}\label{sec:imp}
\twnote{} Our implementation is geared
towards realizing both the level of genericity and the efficiency afforded by
the abstract algorithm, while staying as close as possible at the theory as
presented in the present and preceding work~\cite{concurSpecialIssue}. Regarding genericity, each
basic functor is defined (in its own source file) as a single Haskell data type
that implements two type classes:
\begin{enumerate}
\item the class \lstinline$RefinementInterface$ with functions \texttt{init} and
  \texttt{update}, which directly corresponds to the mathematical notion
  (\autoref{D:refinement-interface}), and
\item the class \lstinline$ParseMorphism$, which provides a parser that defines
  the coalgebra syntax for the functor.
\end{enumerate}
This means that new basic functors can be implemented without modifying any of
the existing code, except for registering the new type in a list of functors
(the existing functor implementations are
in~\href{https://git8.cs.fau.de/software/copar/tree/master/src/Copar/Functors}{\texttt{src/Copar/Functors}}).
The type class modelling refinement interfaces is defined as follows in \copar:
\begin{lstlisting}
class (Ord (F1 f), Ord (F3 f)) => RefinementInterface f where
  init :: F1 f -> [Label f] -> Weight f
  update :: [Label f] -> Weight f -> (Weight f, F3 f, Weight f)
\end{lstlisting}
Here, the type \lstinline$f$ serves as the name of the functor $F$ of interest, and
\lstinline$Label f$ is the type representing the label set $A$ from the encoding
of the functor. Similarly, the type \lstinline$Weight f$ represents $W$ and the
types \lstinline$F1 f$ and \lstinline$F3 f$ represents the sets $F1$ and $F3$.
For example, if we want to implement the refinement interface for $FX=\R^{(X)}$
explicitly we can write the following:
\begin{lstlisting}
data R x = R x
type instance Label R = Double
type instance Weight R = (Double,Double)
type instance F1 R = Double
type instance F3 R = (Double,Double,Double)

instance RefinementInterface R where
   init g e = (0, g)
   update l (r,b) = ((r + b - sum l, sum l),
                     (r,  b - sum l, sum l),
                     (r + sum l, b - sum l))
\end{lstlisting}%
The first line defines a parametrized type \lstinline$R$ (with one constructor
of the same name) representing the functor ($\R^{(-)}$ in this case), with the
parameter \lstinline$x$ representing the functor argument. The next lines define
the types representing the sets that appear in the refinement interface, and the
instance of \lstinline$RefinementInterface$ for \lstinline$R$ implements the
$\op{init}$ and $\op{update}$ routines for $\R^{(-)}$ as we have defined them
before (\itemref{E:refint}{E:refint:2}).

Concerning efficiency, \copar{} faithfully implements our imperative
algorithm~\cite{concurSpecialIssue} in the functional language Haskell. We have
made sure that this implementation actually realizes the good theoretical
complexity of the algorithm. This is achieved by ample use of the
\texttt{ST}~monad~\cite{DBLP:conf/pldi/LaunchburyJ94} and by disabling lazy
evaluation for the core parts of the algorithm using GHC's \texttt{Strict}
pragma. The \texttt{ST}~monad also enables the use of efficient data structures
like mutable vectors where possible.
  
One such data structure, which is central to the efficient implementation of the
generic algorithm, is a \emph{refinable partition}, which stores the blocks of
the current partition of the state set $C$ of the input coalgebra during the
execution of the algorithm. This data structure has to provide constant time
operations for finding the size of a block, marking a state and counting the
marked states in a block. Splitting a block in marked and unmarked states must
only take linear time in the number of marked states of this block. Valmari and
Franceschinis~\cite{ValmariF10} have described a data structure (for use in
Markov chain lumping) fulfilling all these requirements, and this is what we use
in \copar.  

Our abstract algorithm maintains two partitions~$P,Q$ of $C$, where~$P$ is one
transition step finer than~$Q$; i.e.~$P$ is the partition of~$C$ induced by the
map $Fq \cdot c\colon C \to FQ$, where $q\colon C \epito Q$ is the canonical
quotient map assigning to every state the block which contains it. The key to
the low time complexity is to choose in each iteration a \emph{subblock} $S$ in
$P$ whose surrounding \emph{compound block} $B$ in $Q$ (with $S \subseteq B$)
satisfies $2 \cdot |S| \leq |B|$, and then refine $Q$ (and $P$) as explained in
\autoref{sec:refintface} (see \autoref{fig:splitblock}). This idea goes back to
Hopcroft~\cite{Hopcroft71}, and is also used in all other partition refinement
algorithms mentioned in the introduction.  Our implementation maintains a queue
of subblocks $S$ satisfying the above property, and the termination condition
$P = Q$ of the main loop then translates to this queue being
empty. %

One optimization that is new in \copar in relation to previous
work~\cite{ValmariF10,concurSpecialIssue} is that weights for blocks of exactly
one state are not computed, because such blocks cannot be split any
further. This has drastic performance benefits for inputs where the algorithm
produces many single-element blocks early on, e.g.~for nearly minimal systems or
fine-grained initial partitions, see~\cite{Deifel18} for details and
measurements.

\section{Instances}
\label{sec:instances}

Many systems are coalgebras for functors formed according to the
grammar~\eqref{termGrammar2}. In \autoref{tab:instances}, we list various system
types that can be handled by our algorithm, taken from~\cite{concurSpecialIssue}
except for Markov chains with weights in a monoid and weighted tree automata,
which are new in the present paper. In all cases,~$m$ is the number of edges and
$n$ is the number of states of the input coalgebra, and we compare the run time
of our generic algorithm with that of specifically designed algorithms from the
literature. In most instances, we match the complexity of the best known
algorithm. In the one case where our generic algorithm is asymptotically slower
(LTS with unbounded alphabet), this is due to assuming a potentially very large
number of alphabet letters -- as soon as the number of alphabet letters is
assumed to be polynomially bounded in the number~$n$ of states, the number~$m$
of transitions is also polynomially bounded in~$n$, so $\log m\in\CO(\log n)$.
This argument also explains why `$<$' and `$=$', respectively, hold in the last
two rows of \autoref{tab:instances}, as we assume~$\Sigma$ to be (fixed and)
finite; the case where~$\Sigma$ is infinite and unranked is more complicated.
Details on the instantiation to weighted tree automata are discussed in
\autoref{sec:wta}.
\newcommand{\mycite}[1]{\cite{#1}}
\begin{table}
  \caption{Asymptotic complexity of the generic algorithm (2017) compared to
    specific algorithms, for systems with $n$ states and $m$ transitions, respectively
    $m_\Powf$ nondeterministic and $m_\Dist$ probabilistic
    transitions for Segala systems. For simplicity, we assume that $m\ge
    n$ and that $A$ and $\Sigma$
    are finite and fixed. $M$ is a possibly infinite and possibly
    non-cancellative commutative monoid.
    \smallskip}
    \label{tab:instances}
    \setlength\heavyrulewidth{0.25ex}%
    \renewcommand{\arraystretch}{1.2}%
    \normalsize
\renewcommand{\color}[1]{}%
\begin{tabular}{@{}l@{\hspace{4mm}}l@{\hspace{3mm}}c@{\hspace{4mm}}c@{\hspace{4mm}}cll@{}}
    \toprule
    System
    & Functor
    & Run Time
    &
    & Specific algorithm & Year & Reference
    \\
    \toprule
  \makecell[l]{
  DFA
  }
    & $2\times(-)^A$
    & $n\cdot \log n$
    & {$\bm=$}
    & $n\cdot \log n$
    & 1971
    & \mycite{Hopcroft71}
    \\
    \midrule
    \multirow{2}{*}{\makecell[l]{Transition\\ Systems}}
    & \multirow{2}{*}{$\Powf$}
    & \multirow{2}{*}{$m\cdot \log n$}
    & \multirow{2}{*}{{$\bm=$}}
    & 
      \multirow{2}{*}{$m\cdot \log n$} 
    & \multirow{2}{*}{1987}
    &  \multirow{2}{*}{\mycite{PaigeTarjan87}}
    \\\\
    \midrule
     \multirow{2}{*}{\makecell[l]{Labelled Tran-\\ sition Systems}}
    & \multirow{2}{*}{$\Powf(\N\times -)$}
    & \multirow{2}{*}{\makecell{$m\cdot \log m$}}
    & {$\bm=$}
    & $m\cdot \log m$
    & 2004
    & \mycite{DovierEA04}
    \\
    & 
    & 
    & {$\bm >$}
    & $m\cdot \log n$
    & 2009
    & \mycite{Valmari09}
    \\
    \midrule
    Markov Chains
    & $\R^{(-)}$
    &$m\cdot \log n$
    & {$\bm=$}
    & $m\cdot \log n$
    & 2010
    & \mycite{ValmariF10}
    \\
  \midrule
  Weighted Systems
    & $M^{(-)}$
    &$\mathclap{m\cdot \log n\cdot \log(\min(m,|M|))}$
    &
    &
    &
    &
    \\
    \midrule
    \multirow{2}{*}{\makecell[l]{Simple\\ Segala Systems}}
    & \multirow{2}{*}{\makecell[l]{$\Powf(A\times -)\cdot \Dist$}}
    & \multirow{2}{*}{\makecell{$m_\Dist \cdot \log m_\Powf$}}
    & {$\bm <$}
    & $m\cdot \log n$
    & 2000
    & \mycite{BaierEM00}
    \\
    & 
    & 
    & {$\bm =$}
    & \makecell{$m_\Dist \cdot \log m_\Powf$}
    & 2018
    & \cite{GrooteEA18}
    \\
    \midrule
    \multirow{2}{*}{\makecell[l]{Colour\\Refinement}}
    & \multirow{2}{*}{\makecell[l]{$\Bag$}}
    & \multirow{2}{*}{\makecell[l]{$m\cdot \log n$}}
    & \multirow{2}{*}{\makecell[l]{$\bm=$}}
    & \multirow{2}{*}{\makecell[l]{$m\cdot \log n$}}
    & \multirow{2}{*}{\makecell[l]{2017}}
    & \multirow{2}{*}{\makecell[l]{\cite{BerkholzBG17}}}
    \\\\
  \midrule%
  \multirow{3}{*}{\makecell[l]{
  Weighted \\
  Tree\\ Automata
  }}
    & $M\times M^{(\Sigma (-))}$
    & $m\cdot \log^2 m$
    & $\bm<$
    & $m\cdot n$
    & 2007
    & \cite{HoegbergEA07}
      \\
    & \multirow{2}{*}{\makecell[l]{$M\times M^{(\Sigma (-))}$\\[-1mm] \scriptsize ($M$ cancellative)}}
    & \multirow{2}{*}{\makecell[l]{\!\!$m\cdot \log m$\!\!}}
    & \multirow{2}{*}{\makecell[l]{$\bm =$}}
    & \multirow{2}{*}{\makecell[l]{$m \cdot \log n$}}
    & \multirow{2}{*}{\makecell[l]{2007}}
    & \multirow{2}{*}{\makecell[l]{\cite{HoegbergEA07}}}
    \\\\
    \bottomrule
  \end{tabular}%
  \vspace*{-10pt}
\end{table}%

Our algorithm and tool can handle further system types that arise by combining
functors in various ways. For instance, so-called simple Segala systems are
coalgebras for the functor $\Powf(A\times \Dist(-))$, and are minimized by our
algorithm in time $\CO((m+n)\cdot \log n)$, improving on the best previous
algorithm~\cite{BaierEM00}\twnote{} and matching the
complexity of the algorithm by Groote et al.~\cite{GrooteEA18}. Other type
functors for various species of probabilistic systems are listed in
\cite{BARTELS200357}, including the ones for general Segala systems, reactive
systems, generative systems, stratified systems, alternating systems, bundle
systems, and Pnueli-Zuck systems. Hence, \copar{} provides an off-the-shelf
minimization tool for all these types of systems. 

\begin{remark}[Initial partitions]\label{R:initPart}
Note that in the classical
Paige-Tarjan algorithm \cite{PaigeTarjan87}, the input includes an initial
partition. Initial partitions as input parameters are covered via the
genericity of our algorithm. In fact, initial partitions on $F$-coalgebras are
accommodated by moving to the functor $F'X = \N\times FX$, where the
first component of a coalgebra
\[
  C \xrightarrow{\fpair{c_1,c_2}} \N \times FX
\]
assigns to each state the number of its
block in the initial partition. Under the optimized treatment of the
polynomial functor \mbox{$\N\times(-)$} (\autoref{sec:des}), this
transformation does not enlarge the state space and also leaves the
run time factor $\rifactor(c)$ unchanged~\cite{concurSpecialIssue};
that is, the asymptotic run time of the algorithm remains unchanged
under adding initial partitions. 
\end{remark}

\section{Weighted Transition Systems}
\label{sec:monoidvalued}
We have seen in \autoref{E:refint}\ref{E:refint:2} that weighted systems with
weights in a group easily fit into our framework of generic partition
refinement, since we can use inverses to implement the refinement interface
efficiently. In the following we generalize this by allowing the weighted
systems to be weighted in an arbitrary commutative monoid $M$ that does not
necessarily have inverses.

Systems with weights in a monoid are studied by Klin and Sassone~\cite{KlinS13},
and they show that behavioural equivalence of $M^{(-)}$-coalgebras is precisely
weighted bisimilarity (cf.~\itemref{exManyFunctors}{exManyFunctors:3}). Weighted
transition systems with weights in a monoid also serve as a base for -- and are
in fact a special case of -- weighted tree automata as studied by Högberg et
al.~\cite{HoegbergEA09}, which we will discuss in the next section.

In the following we distinguish between cancellative and non-cancellative
monoids because the respective refinement interfaces for $M^{(-)}$ are
implemented differently, with the interface for a cancellative monoid allowing
for a lower time complexity.

\subsection{Cancellative Monoids}
\label{sec:mon:cancellative}

Recall that a commutative monoid $(M,+,0)$ is \emph{cancellative} if $a+b = a+c$
implies $b = c$.  Clearly, every submonoid of a group is cancellative, for example
$(\N,+,0)$ and $(\Z\setminus\{0\},\cdot\,,1)$.
It is well-known that every cancellative
commutative monoid~$M$ embeds into an abelian group~$G$ via the
standard Grothendieck construction. Explicitly,
\[
  G = (M\times M)/\mathord{\equiv}
  \qquad
  (a_+,a_-)
  \equiv
  (b_+,b_-)
  ~\text{iff}~
  a_++b_- = b_++a_-,
\]
and the group structure is given by the usual component-wise
addition on the product:
\[
  [(a_+,a_-)] + [(b_+,b_-)] = [(a_++b_+,a_-+b_-)]
  \quad\text{  and }\quad
-[(a_+,a_-)] = [(a_-,a_+)].
\]
The embedding of $M$ into $G$ is given by the monoid
homomorphism
\[
  \iota\colon M\to G\qquad \iota(m)= [(m,0)].
\]
Hence, we have in total:
\begin{corollary}
A monoid is cancellative iff it is a submonoid of a group.
\end{corollary}
Informally speaking, the inverses of elements of a cancellative monoid $M$
exist, albeit not within $M$ itself. The embedding $M\to G$ extends to a
component-wise injective natural transformation $\alpha\colon M^{(-)}\to
G^{(-)}$, and therefore, computing behavioural equivalence for $M^{(-)}$ reduces
to that of $G^{(-)}$~\cite[Proposition~2.13]{concurSpecialIssue}. Hence, we can
convert every $M^{(-)}$-coalgebra $c\colon C\to M^{(C)}$ into the
$G^{(-)}$-coalgebra
\[
  C\xrightarrow{~c~} M^{(C)}\xrightarrow{~\alpha_C~} G^{(C)}
\]
and use the refinement interface for
$G^{(-)}$ from \autoref{E:refint}\ref{E:refint:2}, obtaining:
\begin{corollary}
  Let $M$ be a cancellative monoid. Then partition refinement on a weighted
  transition system $c\colon C\to M^{(C)}$ with $n$ states and $m$ transitions
  runs in time $\CO((m+n)\cdot \log n)$.
\end{corollary}
\noindent
Indeed, this is immediate from \autoref{T:comp} for $\rifactor(c) = 1$.
\smnote{}

\subsection{Non-cancellative Monoids}
\label{sec:noncan}
There are monoids that are used in practice that fail to be cancellative, for
example the additive monoid $(\N,\max,0)$ of the tropical semiring.
Assume given a non-cancellative commutative monoid $(M,+,0)$. Then~$M$ does not
embed into a group, so we need a new refinement interface for the type functor
$M^{(-)}$ of $M$-weighted transition systems in our algorithm, rather than being
able to reuse the one for group-valued functors as in the case of cancellative
commutative monoids (\autoref{sec:mon:cancellative}).  The basic idea in the
construction of a refinement interface for $M^{(-)}$ is to incorporate bags of
monoid elements into the weights, and consider subtraction of bags. We implement
this idea as follows.

We use the same encoding of $M^{(-)}$ as for group-valued functors:
\[
  A=M_{\neq 0}=M\setminus\{0\},\text{ and }
  \flat(f) = \{\,(f(x), x) \mid x\in X, f(x)\neq 0\,\}\text{ for }f\in M^{(X)}.
\]
The refinement interface for $M^{(-)}$ has weights
$W = M\times\Bagf(M_{\neq 0})$ and uses weight functions
\begin{equation}\label{eq:weight-noncancellative}
  w(B)(f) = \textstyle\big(\sum_{x\in X\setminus B}f(x),\;
  (m\mapsto\big|\{ x\in B\mid f(x) = m\}\big|) \big)\in M\times\Bagf(M_{\neq 0});
\end{equation}
that is, $w(B)(f)$ returns the total weight of $X\setminus B$
under~$f$ and the bag of non-zero elements of~$M$ occurring in
$f$. The interface functions
$\op{init}\colon M^{(1)}×\Bag M_{\neq 0}\to W$,
$\op{update}\colon \Bag M_{\neq 0} × W \to W× M^{(3)}× W$ use the summation map
$\csum\colon \Bag M\to M$ from~\eqref{eq:csum} and are given by
\begin{align*}
  \op{init}(f, \ell) &= (0, \ell)\\
  \op{update}(\ell, (r, c)) &=
  ((r + \Sigma(c-\ell), \ell),
  (r, \Sigma(c-\ell), \Sigma(\ell)),
  (r + \Sigma(\ell), c - \ell)),
\end{align*}
where for $a, b\in \Bagf Y$, the bag $a-b$ is defined by
\[
  (a-b)(y) = \max(0, a(y) - b(y)).
\]
As for groups, we denote elements of $M^{(3)}$ as triples of elements from~$M$.

\begin{proposition}
  For every commutative monoid $M$, the above functions $\op{init}$ and
  $\op{update}$ define a refinement interface for the functor~$M^{(-)}$.
\end{proposition}
\begin{proof}
  We define the weight functions~$w\colon \Pow X\to (M^{(X)}\to M\times
  \Bagf(M_{\neq 0}))$ as in~\eqref{eq:weight-noncancellative},
  and show that $\op{init}$ and $\op{update}$ then satisfy the two axioms in
  \autoref{D:refinement-interface}. Let $t\in FX$. For the first axiom we
  compute as follows:\smnote{}
  \begin{align*}
    w(X)(t) &= \big(\textstyle\sum_{x\in X\setminus X}t(x),\;
    \overbrace{(m\mapsto\big|\{ x\in X\mid t(x) = m\}\big|)}^{\text{ in }M_{\neq 0}\to \N} \big)
    \\
    &= \big(0, \mbraces{m \mid x \in X, t(x) = m,\,m \neq 0}\big)
    \\
    &= \big(0,\mbraces{m \mid (m,x) \in \{(t(x), x) \mid x \in X, t(x) \neq
      0\}}\big)
    \\
    &= \op{init}\big(F!(t), \mbraces{m \mid (m,x) \in\flat (t)} \big).
  \end{align*}
  In the first of the above multiset comprehensions, different $x\in X$ with $t(x) =m, m\neq
  0$ lead to multiple occurrences of $m$ in the multiset, and similarly in the second
  multiset comprehension.
  Let us now verify the second axiom concerning $\op{update}$. In order to
  simplify the notation, we define 
  \[
    (t\downarrow B) \in \BagM,
    \quad
    (t\downarrow B)(m) = |\{x\in B\mid t(x) = m\}|
    \quad
    \text{for }B\subseteq X, t\in M^{(X)}.
  \]
  The accumulated weight of edges into $B\subseteq X$ in $t\in M^{(X)}$ is
  then denoted by \newcommand{\xsum}[1]{\underset{#1}{{\textstyle\sum\,}}}
  \[
    \xsum{B} t := \csum(t\downarrow B) = \textstyle\sum_{x\in B} t(x).
  \]
  In this notation, we have
  \begin{align*}
    w(B)(t)
    &= \big(\textstyle\sum_{x\in X\setminus B}t(x), m\mapsto |\{x\in B\mid t(x) = m\}|\big)
    \\
    &= (\xsum{X\setminus B}t, (t\downarrow B)).
  \end{align*}
  With subtraction of bags defined as above by $(a-b)(y) = \max(0, a(y)-b(y))$
  for $a,b\in \Bag Y$, we have
  \[
    (t\downarrow B) - (t\downarrow S)
    = (t\downarrow (B\setminus S))
    \qquad\text{for }S\subseteq B\subseteq X.
  \]
  Then we compute as follows, for $S\subseteq B\subseteq X$:
  \allowdisplaybreaks
  \begin{align*}
    &\phantom{=}\,\,\,\fpair{w(S),F\chi_S^B,w(B\!\setminus\! S)}(t)
      \\
    &= (w(S)(t),F\chi_S^B(t),w(B\!\setminus\! S)(t))
      \\
    &= \big((\xsum{X\setminus S}t,(t\downarrow S)),~
      (\xsum{X\setminus B} t, \xsum{B\setminus S}t, \xsum{S}t),~
      (\xsum{X\setminus (B\setminus S)}t,(t\downarrow (B\setminus S)))\big)
    \\
    &=
      \begin{array}[t]{rll}
      \big(
        & (\xsum{X\setminus B} t + \smash{\underbrace{\csum(t\downarrow (B \setminus S))}_{
          \xsum{B\setminus S}t}}
          , (t\downarrow S)),
      &(\xsum{X\setminus B} t, \csum(t\downarrow (B\setminus S)), \csum(t\downarrow S)),\\
      && (\xsum{X\setminus B} t + \xsum{S} t, ((t\downarrow B) - (t\downarrow S)))
      \big)
      \end{array}
    \\
    &=
      \begin{array}[t]{r@{}l}
      \big(
        & (\xsum{X\setminus B} t + \csum((t\downarrow B) - (t\downarrow S)), (t\downarrow S)), \\
      &(\xsum{X\setminus B} t, \csum((t\downarrow B) - (t\downarrow S)), \csum(t\downarrow S)),\\
      & (\xsum{X\setminus B} t + \csum(t\downarrow S), ((t\downarrow B) - (t\downarrow S)))
      \big)
      \end{array}
    \\
    &\overset{(*)}{=} \op{update}((t\downarrow S), (\xsum{X\setminus B} t, (t\downarrow B)))
    \\
    &= \op{update}\big((m\in M_{\neq 0}\mapsto |\{x\in S\mid t(x) = m\}|), w(B)(t)\big)
      \\
    &= \op{update}\big(\mbraces{m\in M_{\neq 0}\mid (m,x)\in \flat(t), x\in  S}, w(B)(t)\big)\\
    &= \op{update}\big(\mbraces{a \in A\mid (a,x)\in \flat(t), x\in S}, w(B)(t)\big),
  \end{align*}
  where the step labelled~$(*)$ uses the definition of $\op{update}$:
  \[
    \op{update}(\ell, (r, c)) =
    ((r + \Sigma(c-\ell), \ell),
    (r, \Sigma(c-\ell), \Sigma(\ell)),
    (r + \Sigma(\ell), c - \ell)).
    \tag*{\proofbox}
  \]
  \let\proofbox\relax
\end{proof}
\noindent In order to determine the run time factor of the above refinement
interface, we need to describe how we handle elements of
$W=M\times\Bagf(M_{\neq 0})$ in the routines of the refinement interface. We
implement a bag in $\Bagf(M_{\neq 0})$ as a balanced search tree with keys
$M_{\neq 0}$ and values $\N$.  In addition to the standard structure of a
balanced search tree, we store in every node the value $\Sigma(b)$, where $b$ is
the bag encoded by the subtree rooted at that node. Hence, for every bag
$b\in\Bagf(M_{\neq 0})$, the value $\Sigma(b)$ is immediately available at the
root node of the search tree for $b$. For \copar, we have implemented the basic
operations on balanced search trees following Adams~\cite{Adams93}. For the
complexity analysis, we prove that maintaining the values $\Sigma(b)$ in the
nodes only adds constant overhead to the operations on search trees, so that
we obtain
\begin{proposition}\label{P:monoids-fast}
  The above refinement interface for $M^{(-)}$ has run time factor
  \[
    \rifactor(n,m) = \log\min(|M|,m).
  \]
\end{proposition}
\noindent
More precisely, the above functions $\op{init}(f, \ell)$ and
$\op{update}(\ell, (r,c))$ can be computed in time
$\CO(|\ell|\cdot\log\min(|M|, m))$, where $m$ is the number of edges in the
input coalgebra, and values in $M^{(3)}$ can be compared in constant time.
\begin{proof}
  For a node $x$ in a binary search tree encoding a bag in $\Bag(M_{\neq
    0})$ as described above, we write $\Sigma(x)$ for the value in $M$ stored at that
  node.\smnote{}
  
  Note that our search trees cannot have more nodes than the size $|M|$ of
  their index set, and the number of nodes is also not greater than
  the number $m$ of all edges. Hence, their size is bounded by
  $\min(|M|, m)$.

  Recall from algorithm textbooks (e.g.~\cite[Section~14]{CormenEA}) that
  the key operations $\op{insert}$, $\op{delete}$ and $\op{search}$ have
  logarithmic time complexity in the size of a given balanced binary search
  tree.

  We need to argue that maintaining the values $\Sigma(x)$ in the nodes does not
  increase this complexity. This is obvious for the $\op{search}$ operation as
  it does not change its argument search tree. For $\op{insert}$ and
  $\op{delete}$, \lsnote{} these
  operations essentially trace down one path starting at the root to a node (at
  worst, a leaf) of the given search tree (as described e.g.\ by~\cite{CormenEA}).
  Additionally, we need to rebalance the search tree after inserting or deleting
  a node. This is done by tracing back the same path to the root and (possibly)
  performing \emph{rotations} on the nodes occurring on that path. Rotations are
  local operations changing the structure of a search tree but preserving the
  inorder key ordering of subtrees (see \autoref{fig:rotations}).
  \begin{figure}
    \centering
    \begin{tikzpicture}
      \begin{scope}[searchtree, only math nodes]
        \node[root] (t1) {}
        child { node (y1) {y}
          child { node (x1) {x}
            edge from parent[draw=none]
            child { node[subtree] {\alpha} }
            child { node[subtree] {\beta} }
          }
          child { node[subtree] (gamma1) {\gamma} }
        } ;
        \path[draw] (y1) to (x1);
      \end{scope}
      \begin{scope}[searchtree, only math nodes]
        \node[root] (t2) at (5cm,0) {}
        child { node (x2) {x}
          child { node[subtree] (alpha2) {\alpha} }
          child { node (y2) {y} edge from parent[draw=none]
            child { node[subtree] {\beta} }
            child { node[subtree] {\gamma} }
          }
        } ;
        \path[draw] (y2) to (x2);
      \end{scope}
      \coordinate (lefttree) at ([yshift=2mm]gamma1.south east);
      \coordinate (righttree) at (alpha2.south west |- lefttree);
      \begin{scope}[draw,->,shorten <= 2.2mm,shorten >= 2mm]
      \path
      ([yshift=3mm]lefttree) edge node[above]{right rotation} ([yshift=3mm]righttree)
      ([yshift=1mm]righttree) edge node[below]{left rotation} ([yshift=1mm]lefttree)
      ;
      \end{scope}
    \end{tikzpicture}
    \caption{Rotation operations in binary search trees.}\label{fig:rotations}
  \end{figure}
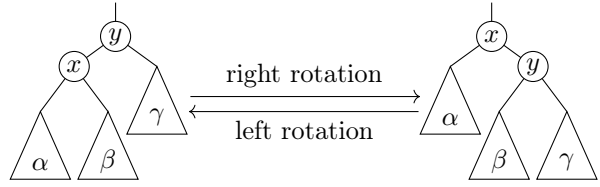
  
  Clearly, in order to maintain the correct summation values in a
  search tree under a rotation, we only need to adjust those values in
  the nodes $x$ and $y$. This is achieved as follows: 
  \begin{align*}
    \Sigma(x) &= \Sigma(\alpha) + \Sigma(\beta);
    \qquad
    \Sigma(y) = \Sigma(x) + \Sigma(\gamma)
    &
    \text{for left rotation},
    \\
    \Sigma(y) &= \Sigma(\beta) + \Sigma(\gamma);
    \qquad
    \Sigma(x) = \Sigma(\alpha) + \Sigma(y)
    &\text{for right rotation}.
  \end{align*}
  In addition, when inserting or deleting a node $x$ we must recompute the
  $\Sigma(y)$ of all nodes $y$ along the path from the root to $x$ when we trace
  that path back to the root. This can clearly be performed in constant time for
  each node $y$ since $\Sigma(y)$ is the sum (in $M$) of $\Sigma(y_1)$ and
  $\Sigma(y_2)$, which are stored at the child nodes~$y_1$ and~$y_2$ of~$y$,
  respectively.

  In summary, we see that maintaining the desired summation values only requires
  an additional constant overhead in the backtracing step. Consequently, the
  operations of our balanced binary search trees run in time
  $\CO(\log \min(|M|,m))$. Subtraction $c - \ell$ of bags $c,\ell$ performs
  $|\ell|$-many calls to $\op{delete}$ on~$c$, and computing the sum
  $\csum(\ell)$ takes time $\CO(|\ell|)$, since $\ell$, the bag of \emph{labels}
  passed to $\op{update}$, is not represented as a search tree.  Hence, we
  obtain the desired overall time complexity $\CO(|\ell|\cdot \log \min(|M|,m))$
  of $\op{update}$.

  Similarly, for $\op{init}(f,\ell)$, we need $|\ell|$-many calls to
  $\op{insert}$ in order to initialize the search tree
  representing~$\Bag (M_{\neq 0})$.
\end{proof}
\begin{remark}
  It is no coincidence that we use $\Bag (M_{\neq 0})$ in the refinement
  interface. In fact, for every set $X$, the set $\Bag X$, with union of bags as
  addition, is the free commutative monoid on $X$. Moreover, $\Bag X$ is
  cancellative, so that we may use a form of subtraction on bags. Thus, we see
  that $\Bag(M_{\neq 0})$ is a canonical cancellative monoid containing
  $M_{\neq 0}$ (via the identification of elements of $M_{\neq 0}$ with
  singleton bags). Moreover, the summation map
  \[
    \csum\colon \Bag M\to M
  \]
  is the canonical unique monoid homomorphism freely extending the identity map
  on $M$. Thus, this map allows us to go back from bags to monoid elements.
  \smnote{}\lsnote{}  This is essentially the point of Eilenberg-Moore algebras in
  general (cf.~\autoref{R:EilenbergMoore}).
\end{remark}

\begin{corollary}
  Let $M$ be any commutative monoid. Then partition refinement on  weighted
  transition systems $c\colon C\to M^{(C)}$ with $n$ states and $m$ transitions
  runs in time $\CO((m+n)\cdot \log n\cdot\log\min(|M|,m))$.
\end{corollary}
\noindent
Indeed, this is immediate by \autoref{P:monoids-fast} and \autoref{T:comp}.

\section{Weighted Tree Automata}
\label{sec:wta}
We proceed to take a closer look at weighted tree automata as a worked example.
It is this example that mainly motivates the discussion of
non-cancellative monoids in the last section, since in this case the generic
algorithm improves on the run time of the best known specific algorithms in the
literature.

Weighted tree automata simultaneously generalize tree automata and
weighted (word) automata. A partition refinement construction for
weighted automata (w.r.t.~weighted bisimilarity) was first considered
by Buchholz~\cite[Theorem~3.7]{Buchholz08}. Högberg et al.\ first
provided an efficient partition refinement algorithm for tree
automata~\cite{HoegbergEA09}, and moreover for weighted tree
automata~\cite{HoegbergEA07}. Generally, tree automata differ from
word automata in replacing the input alphabet, which may be seen as
sets of unary operations, with an algebraic signature~$\Sigma$:

\begin{definition}\label{D:wta}
  Let $(M,+,0)$ be a commutative monoid. A (bottom-up) \emph{weighted tree
    automaton} (WTA) (over $M$) consists of a finite set $X$ of
  states, a finite signature $\Sigma$, an output map $f\colon X\to M$,
  and for each $k\ge 0$, a transition map
  $\mu_k\colon \Sigma_k\to M^{X^k\times X}$, where $\Sigma_k$ denotes
  the set of $k$-ary input symbols in $\Sigma$; the maximum arity of
  symbols in~$\Sigma$ is called the \emph{rank}.
\end{definition}

\noindent%
Given a weighted tree automaton $(X, f, (\mu_k)_{k \in \N})$ as in
\autoref{D:wta} we see that it is, equivalently, a finite coalgebra
for the functor $FX = M \times M^{(\Sigma X)}$, where we identify the
signature $\Sigma$ with its
corresponding polynomial functor $\Sigma X = \coprod_{\arity[\scriptstyle]{\sigma}{k} \in
  \Sigma} X^k$. Indeed, $(\mu_{k})_{k\ge 0}$ is equivalently
expressed by a map
\begin{equation}
  \bar \mu \colon X\to M^{(\Sigma X)}
  \quad
  \text{with}
  \quad
  \bar \mu (x)(\sigma(x_1,\ldots,x_k))
  := \mu_k(\sigma)((x_1,\ldots,x_k), x).
  \label{eqMuBar}
\end{equation}
Note that $\bar \mu(x)$ is finitely supported because $X$ and $\Sigma
X$ are finite.
Thus we obtain a coalgebra
\begin{equation}\label{coalgWTADef}
  g\colon X\to M\times M^{(\Sigma X)}
  \quad
  \text{with}
  \quad
  g(x) = (f(x), \bar \mu(x)).
\end{equation}

\begin{example}\label{ex:wta}
  Let $\Sigma = \{\arity{*}{2}\}$ be a signature with a single binary symbol.
  Consider the WTA $(X, f, \mu_2)$ for~$\Sigma$ over the monoid
  $(\Z,\max,-\infty)$ (simply $(\Z,\max)$ in the following)
  with $X= \{a,b,c,d\}$ and $f\colon X\to\Z$ being the constant function $x\mapsto 1$ and $\mu_2$
  given by
  \begin{alignat*}{4}
    \mu_2(*)(b, a, a) &= 3,\quad & \mu_2(*)(a, a, a) &= 5,\quad &\mu_2(*)(a, b, b) &= 5,\quad & \mu_2(*)(b, b, b) &= 2,\\
    \mu_2(*)(b, a, c) &= 5,\quad & \mu_2(*)(c, a, c) &= 7,\quad &\mu_2(*)(b, a, d) &= 5,\quad & \mu_2(*)(a, c, d) &= 7\quad\text{ and}\\
    \mu_2(*)(x, y, z) &= -\infty \quad &
    \text{otherwise.} %
  \end{alignat*}
  This WTA is equivalently expressed as the coalgebra
  $g\colon X\to\Z\times(\Z,\max)^{(X\times X)}$ given by
  $g(x) = (f(x), \bar\mu(x))$ with $\bar\mu\colon X\to (\Z,\max)^{(X\times X)}$
  defined as
  \begin{alignat*}{4}
    \bar\mu(a)(b, a) &= 3,\quad & \bar\mu(a)(a, a) &=5,\quad & \bar\mu(b)(a, b) &= 5,\quad & \bar\mu(b)(b, b) &= 2, \\
    \bar\mu(c)(b, a) &= 5, & \bar\mu(c)(c, a) &=7, & \bar\mu(d)(b, a) &= 5, & \bar\mu(d)(a, c) &= 7,
  \end{alignat*}
  where again $\bar\mu$ is $-\infty$ in all other cases. In the syntax of our tool,
  the map $\bar \mu\colon X\to \Z^{(X\times X)}$ can be written as:
\begin{verbatim}
(Z,max)^(X×X)
a: {(b, a): 3, (a, a): 5}
b: {(a, b): 5, (b, b): 2}
c: {(b, a): 5, (c, a): 7}
d: {(b, a): 5, (a, c): 7}
\end{verbatim}
\end{example}

\medskip\noindent%
For the minimization of weighted tree automata, two bisimulation notions are
considered in the literature~\cite{Buchholz08,HoegbergEA07}: \emph{forward} and
\emph{backward} bisimulation. Here, we treat backward
bisimulation, as it corresponds to coalgebraic behavioural equivalence.

\begin{definition}[Högberg et al.~{\cite[Def.~16]{HoegbergEA07}}]
  A \emph{backward bisimulation} on a weighted tree automaton
  $(X, f, (\mu_k)_{k \in \N})$ is an equivalence relation
  $R\subseteq X \times X$ such that for every $(p,q) \in R$, $\sigma/k\in \Sigma$, and
  every $L\in \{D_1\times\cdots\times D_k\mid D_1,\ldots,D_k\in X/R\}$ the
  following equation holds:
  \begin{equation}\label{eq:sum}
    \sum_{w\in L}
    \mu_k(\sigma)(w,p)
    =
    \sum_{w\in L}
    \mu_k(\sigma)(w,q).
  \end{equation}
\end{definition}
\begin{remark}\label{rem:w-in-L}
  Note that $L$ consists of the $w\in X^k$ such that $e^k(w) = L$, where
  $e^k\colon X^k\epito (X/R)^k$ is the $k$-fold power of the canonical quotient
  map $e\colon X \to X/R$.
\end{remark}
\begin{example}
  The equivalence relation $R = \{a,b\}^2\cup \{c\}^2\cup \{d\}^2$ is a backward
  bisimulation for the automaton defined in \autoref{ex:wta}.
\end{example}
\noindent
We can regard the output map $f$ as a transition map for a constant symbol, so
it suffices to consider the functor $FX = M^{(\Sigma X)}$ (and in fact the
output map is ignored in the definition of backward bisimulation given above and
in \cite{HoegbergEA07}). Then, we obtain the following result:
\begin{proposition}\label{wtaBackward} Backward bisimulation of
  weighted tree automata coincides with behavioural equivalence of
  $M^{(\Sigma(-))}$-coalgebras.
\end{proposition}
\begin{proof}
  We show that for every $M^{(\Sigma(-))}$-coalgebra
  $\bar \mu\colon X\to M^{(\Sigma X)}$ defined as in \eqref{eqMuBar}, an
  equivalence relation $R\subseteq X \times X$ is a backward bisimulation on the
  corresponding WTA iff the canonical quotient map
  $e\colon X\twoheadrightarrow X/R$ is an $M^{(\Sigma(-))}$-coalgebra
  homomorphism with domain $(X,\bar \mu)$. First, let
  $x\in X, \arity{\sigma}{k}\in \Sigma$, $D_1,\dots,D_k\in X/R$, and
  $L=D_1\times\dots\times D_k$ for some equivalence relation $R$. Then we have
  $\sigma(D_1,\dots,D_k) \in \Sigma(X/R)$, and recalling
  Remark~\ref{rem:w-in-L}, we can rewrite the sums in~\eqref{eq:sum} in terms of
  the map $M^{(\Sigma e)}\colon M^{(\Sigma X)}\to M^{(\Sigma(X/R))}$ as follows:
  \begin{align*}
    \sum_{w\in L} \mu_k(\sigma)(w,x)
    &= \sum_{\mathclap{\substack{w\in X^k\\e^k(w) =L}}} \mu_k(\sigma)(w,x)
    = \sum_{\mathclap{\substack{w\in X^k\\\tau\in \Sigma_k\\e^k(w) =L\\\tau=\sigma}}} \mu_k(\tau)(w,x)
    = \sum_{\mathclap{\substack{w\in X^k\\\tau\in \Sigma_k\\\Sigma e(\tau (w)) =\sigma(L)}}} \mu_k(\tau)(w,x)
    =
    \sum_{\mathclap{\substack{\tau(w)\in \Sigma X\\
    \Sigma e(\tau(w)) =\sigma(L)}}} \mu_k(\tau)(w,x)
    \\ & \overset{\mathclap{\text{\eqref{eqMuBar}}}}{=}~~
    \sum_{\mathclap{\substack{\tau(w)\in \Sigma X\\
    \Sigma e(\tau(w)) =\sigma(L)}}} \bar \mu(x)(\tau(w))
    =
    \sum_{\mathclap{\substack{t\in \Sigma X\\
          \Sigma e(t) =\sigma(L)}}} \bar \mu (x)(t)
    = M^{(\Sigma e)}(\bar\mu(x))(\sigma(L)).
  \end{align*}
  Hence, for every equivalence relation $R\subseteq X\times X$ we have the
  following chain of equivalences:
  \begin{align*}
    &\text{$R$ is a backward bisimulation}
    \\
    \Leftrightarrow~
    & \forall (p,q)\in R, \arity{\sigma}{k}\in \Sigma, L\in (X/R)^k\colon&
    \sum_{w\in L} \mu_k(\sigma)(w,p) &=\sum_{w\in L} \mu_k(\sigma)(w,q)
    \\
    \Leftrightarrow~
    & \forall (p,q)\in R, \arity{\sigma}{k}\in \Sigma, L\in (X/R)^k\colon&
    M^{(\Sigma e)}(\bar\mu(p))(\sigma(L))
    &= M^{(\Sigma e)}(\bar\mu(q))(\sigma(L))
    \\
    \Leftrightarrow~
    & \forall (p,q)\in R, t\in \Sigma(X/R)\colon&
      M^{(\Sigma e)}(\bar\mu(p))(t)
      &= M^{(\Sigma e)}(\bar\mu(q))(t)
      \\
    \Leftrightarrow~
    & \forall (p,q)\in R\colon &
      M^{(\Sigma e)}(\bar\mu(p)) &= M^{(\Sigma e)}(\bar\mu(q))
      \\
    \Leftrightarrow~
    & \forall (p,q)\in R\colon &
      (M^{(\Sigma e)}\cdot\bar\mu)(p) &= (M^{(\Sigma
        e)}\cdot\bar\mu)(q)
    \end{align*}
    The last equation holds iff there
    exists a map $r\colon X/R\to M^{(\Sigma(X/R))}$ such that $r\cdot e
    = M^{(\Sigma e)}\cdot\bar\mu$, that is,
    iff $e$ is a coalgebra homomorphism:
    \[
      \begin{tikzcd}[row sep = 7mm,baseline=(B.base)]
        X
        \arrow{r}{\bar \mu}
        \arrow[->>]{d}[swap]{e}
        &
        M^{(\Sigma X)}
        \arrow[->>]{d}{~M^{(\Sigma e)}}
        \\
        X/R
        \arrow[dashed]{r}{r}
        &
        |[alias=B]| M^{(\Sigma (X/R))}
      \end{tikzcd}
      \tag*{\proofbox}
    \]
    \let\proofbox\relax
\end{proof}

\begin{example}
  As we mentioned already, weighted tree automata subsume two other notions:
  \begin{enumerate}
  \item For $M$ being the Boolean monoid (w.r.t.~\emph{or}), we obtain
    (bottom-up) non-deterministic tree automata and their bisimilarity.
  \item If all operations $\sigma \in \Sigma$ have arity one, then
    we obtain ordinary weighted automata and weighted bisimilarity.
  \end{enumerate}
\end{example}

\noindent Since $M^{(\Sigma (-))}$ is composed of~$M^{(-)}$ and a polynomial
functor~$\Sigma$, we have all the refinement interfaces defined already in
previous work and in \autoref{sec:monoidvalued}, distinguishing like Högberg et
al.~\cite{HoegbergEA07} between cancellative and non-cancellative monoids and
obtaining run time factors of $\rifactor(n, m) = 1$ and
$\rifactor(n, m) = \log\min(|M|, m)$ respectively. Both cases appear in the
applications of weighted tree automata to natural language processing, for example the
cancellative monoid of real numbers with addition
$(\R,+,0)$~\cite{PetrovEA06Learning,PetrovKlein07Improved}, or the
non-cancellative monoid $(\N,\max,0)$ from the tropical semiring~\cite{MayKnight06}.

In the following, we
additionally distinguish between finite and infinite monoids.\hpnote{}

\begin{theorem}\label{T:cancel}
  Let $M$ be a commutative monoid. On weighted tree automata with~$n$
  states,~$k$ transitions, and rank~$r$, our algorithm runs in time
  \begin{enumerate}
  \item\label{T:cancel:1} \(\CO((r^2k + r n)\cdot \log(k+n))\) if $M$ is cancellative or finite,
    and
  \item \(\CO\big((rk +n)\cdot \log(k+n) \cdot (\log k + r)\big)\) otherwise.
  \end{enumerate}
\end{theorem}
\begin{proof}
  \takeout{}%
The functor $FX = M^{(\Sigma X)}$ is first transformed into
$F'X = M^{(X)} + \Sigma X$ according to \autoref{sec:des}. Given a coalgebra
$c\colon C \to M^{(\Sigma C)}$ with $n = |C|$ states, this transformation
introduces a set~$K$ of intermediate states, one for every outgoing transition
from every $x \in C$:
\[
  K = \{(x,t) \in C\times \Sigma C \mid c(x)(t) \neq 0\},
\]
hence $|K| = k$. The given coalgebra structure yields the two
evident maps $c_1\colon C \to M^{(K)}$ and $c_2\colon K
\to \Sigma C$ given by
\[
  c_1(x)(x',t) =
    \begin{cases}
      c(x)(t) & \text{if $x' = x$} \\
      0 & \text{otherwise,}
    \end{cases}
         \qquad 
         \qquad 
  c_2(x,t) = t.
\]
It takes $k$ edges to encode $c_1$ and at most $k\cdot r$ edges to encode $c_2$.
Partition refinement is now performed on the following
$F'$-coalgebra:
\[
  C + K \xrightarrow{c_1 + c_2} M^{(K)} + \Sigma C
  \xrightarrow{M^{(\inr)} + \Sigma\inl} M^{(C+K)} + \Sigma(C+K),
\]
where $\inl\colon C\to C+K$ and $\inr\colon K\to C+K$ are the canonical
injections. This coalgebra on $C+K$ has $n' := |C + K| = n + k$ states and at
most $m'= (r+1)\cdot k$ edges. Since the refinement interface for $F'$ is a
combination of those of $M^{(-)}$ and $\Sigma$, its run time factor is given by
the maximum of the run time factors $p_M(c_1)$ and $p_\Sigma(c_2)$,
respectively, of those two refinement interfaces
(cf.~\autoref{R:rifactorSortMax}). We can further simplify this maximum to the
asymptotically equivalent\hpnote{} sum
\[
  \rifactor(c) = \rifactor_M(c_1) + \rifactor_\Sigma(c_2).
\]
Since $\rifactor_\Sigma(c_2) = r$ and the number of edges in $c_1$ is bounded by
$k$, we obtain, by \autoref{T:comp}, an overall time complexity of
\[
  \CO((m'+n') \cdot \log n' \cdot (\rifactor_M(c_1)+\rifactor_\Sigma(c_2)))
  = \CO((r \cdot k + n)\cdot \log (n+k) \cdot (\rifactor_M(c_1) + r)).
\]
We proceed by distinguishing the following cases:
\begin{enumerate}[label=(\alph*)]
  \item If $M$ is cancellative, then we can use the refinement interface for groups
    (see \autoref{E:refint}\ref{E:refint:2})
    with $\rifactor_M(c_1) = 1$ as explained in \autoref{sec:mon:cancellative}.
    Thus, the overall time complexity simplifies to
    \[
      \CO((r \cdot k + n)\cdot \log (n+k) \cdot r)
      = \CO((r^2 \cdot k + r\cdot n)\cdot \log (n+k)).
    \]
  \item Otherwise, we have $\rifactor_M(c_1) = \log \min(|M|, k)$ by
    \autoref{P:monoids-fast}, since $c_1$ has at most $k$ edges.
    \begin{enumerate}[label=(b\arabic*)]
    \item If $M$ is finite, then we have
      $\rifactor_M(c_1) \le \log |M| \in \CO(1)$, and thus obtain the same
      overall run time complexity as in the previous case. We have thus proved
      item~\ref{T:cancel:1} in the statement of the theorem.
    \item If $M$ is infinite, then $\log\min(|M|,k) =\log k$. Thus,
      the overall run time is in
      \[
        \CO((r \cdot k + n)\cdot \log (n+k) \cdot (\log k + r)). \tag*{\proofbox}
      \]
    \end{enumerate}
  \end{enumerate}
  \let\proofbox\relax %
\end{proof}
\vspace*{-10pt}%
\noindent
Note that the number $m$ of edges of the input coalgebra satisfies $m \leq
rk$. Thus, for a fixed input signature~$\Sigma$, we see that $m$ and $k$ are
asymptotically equivalent. If we further assume that $m \geq n$, which means
that there are no isolated states, then we obtain the bound in
\autoref{tab:instances}:
\begin{corollary}\label{C:table}
  Let $M$ be a commutative monoid. For a fixed input signature and on input
  coalgebras with~$n$ states and $m\ge n$ edges, our algorithm runs in time 
  \begin{enumerate}
  \item $\CO(m \cdot \log(m))$, if $M$ is cancellative, and
  \item $\CO(m \cdot \log(m)^2)$, otherwise.
  \end{enumerate}
\end{corollary}
\begin{remark} We now provide a comparison of the
  complexity of Högberg et al.'s algorithm with
  the instances of our algorithm for weighted tree automata.
  \begin{enumerate}
  \item For arbitrary (non-cancellative) monoids, Högberg et al.\ establish a
    complexity of $\CO(r \cdot k \cdot n)$~\cite[Theorem~27] {HoegbergEA07}.
    Under the assumptions of \autoref{C:table}, we see that our bound indeed
    improves the complexity $\CO(r\cdot k\cdot n) = \CO(m \cdot n)$ of Högberg
    et al.'s algorithm. To see this, note first that the number $m$ of edges is
    in $\CO(n^{r+1})$, so that we obtain
    \[
      \CO(m\cdot \log(m)^2)
      \subsetneq
      \CO(m\cdot \sqrt[r+1]{m})
      \subseteq
      \CO(m\cdot \sqrt[r+1]{n^{r+1}})
      = \CO(m\cdot n)
    \]
    using in the first step that $\CO(\log(m)^d)\subsetneq \CO(m^{c})$ for every
    $d \ge 1$ and $0 < c < 1$.
  \item For cancellative monoids, the time bound given by Högberg et al.\ is
    $\CO(r^2\cdot k\cdot\log n)$ \cite[Theorem~29]{HoegbergEA07}. Assuming again
    that $m\ge n$, and recalling that $rk \geq m$, the complexity of our
    algorithm according to \autoref{T:cancel} is
    $\CO(r^2 \cdot k \cdot \log(k+n))$, i.e.\ only slightly worse for
    non-constant signatures.
  \end{enumerate}
\end{remark}

\noindent In addition to guaranteeing a good theoretical complexity, our tool
immediately yields an efficient implementation. For the case of
non-cancellative monoids, this is, to the best of our knowledge, the
only available implementation of partition refinement for weighted
tree automata.

\section{Evaluation and Benchmarking}\label{sec:bench}
We report on a number of benchmarks\footnote{The full set of benchmarks and
  their results can be found at
  \url{https://git8.cs.fau.de/software/copar-benchmarks}} that illustrate the
practical scalability of our tool \copar{} and hence our generic algorithm.
These benchmarks cover a selection of different system types and include randomly
generated inputs as well as real world examples. We also compare \copar{} with
two other minimization tools, where applicable. Details and results of further
benchmarks, in particular for the optimizations described in
\autoref{sec:combprod} and at the end of \autoref{sec:imp}, are reported
in~\cite{Deifel18}. All benchmarks were run and measured on the same Intel®
Core™ i5-6500 processor with 3.20GHz clock rate running a Linux system. We
report the timing results of our tool \copar{} (compiled with GHC 8.4.4)
separately for the three phases parsing, initialization and the actual
refinement loop.

Recall from \autoref{R:initPart}, that the input coalgebra implicitly defines an
initial partition according to the output behaviour of states. We have taken
care to ensure that in all the following benchmarks, this initial partition is
still coarse, i.e.~the algorithm has to perform some actual refinement steps
after initialization.

\subsection{Maximal Feasible Weighted Tree Automata}
We first focus on the instantiation of our algorithm for weighted tree automata
as described in \autoref{sec:wta}. Previous studies on the practical performance of
partition refinement on large labelled transition
systems~\cite{valmari2010simple,Valmari09} show that memory rather than run time
seems to be the limiting factor. Since labelled transition systems are a special
case of weighted tree automata, we expect to see similar phenomena. Hence, we
evaluate the maximal automata sizes that can be processed on a typical current
computer setup: We randomly generate weighted tree automata for various
signatures and monoids, looking for the maximal size of WTAs that can be handled
with 16 GB of RAM, and measure the respective run times of our tool.
\twnote{}
To this end, we minimize randomly generated coalgebras for the functors
\begin{equation}\label{eq:F}
  FX = M\times M^{(\Sigma X)},\quad \text{with $\Sigma X= 4\times X^r$},
\end{equation}
where $r$ ranges over $\{1,\ldots, 5\}$ and weight monoids~$M$ range over

\smallskip\noindent
\begin{tabular}{c@{\hspace{1em}}ll}
$\bullet$ & $(2,\vee,0)$ & (functor available as powerset \texttt{P(X)} in
$\copar$) \\[3pt]
$\bullet$ & $(\N,\max,0)$ & (syntactically: \texttt{(Z,max)\textasciicircum(X)})\\[3pt]
$\bullet$ & $(2,\vee,0)^{64}\cong (\Powf(64),\cup,\emptyset)$  & (syntactically: \texttt{(Word,or)\textasciicircum(X)})
\end{tabular}

\smallskip\noindent

\noindent We write~$n$ for the number of states, $k$ for the number of
transitions, and~$m$ for the number of edges in the coalgebra
encoding.\smnote{}  When
generating a coalgebra with $n$ states, we randomly create 50 outgoing
transitions per state, obtaining $k=50\cdot n$ transitions in total. The
transformation described in \autoref{sec:des} additionally introduces one
intermediate state per transition, leading to an actual number $n' = 51\cdot n$
of states. Every transition of rank~$r$ has one incoming edge and $r$ outgoing
edges, hence $m=(r+1)\cdot k = 50\cdot(r+1)\cdot n$.

\autoref{tab:extended} lists the maximal sizes of weighted tree automata that
\copar{} is able to process in the mentioned 16~GB of RAM, along with the
associated run times. Since our implementation of the refinement interface for
$\Powf\cong (2,\vee,0)^{(-)}$ is optimized for its specific functor, the tool
needs less memory in this case, allowing for higher values of~$n$, an effect
that decreases with increasing rank~$r$.\hpnote{}

When generating weighted tree automata for partition refinement, one needs to be
careful to avoid systems for which the partition refinement is trivial -- either
because the `init' step already distinguishes all states or because `refine'
does not split any blocks. To this end, we restrict to generating at most 50
different elements of $M$ in each automaton, and we generate 50 transitions per
state. As one can tell from the columns $P_1'$ and $P_f'$ in
\autoref{tab:extended}, this strategy excludes the above-mentioned trivial
cases: the initial partition $P_1'$ is small compared to the total number of
states $n'$ whereas the size $P_f'$ of the final partition nearly matches $n'$.
This implies that the actual partition refinement process has to do the main
work of distinguishing nearly all of the $n'$ states; hence, the test cases
present the worst case when it comes to the run time of the algorithm.

The first refinement step produces in the order of $|\Sigma|\cdot \min(50,
|M|)^r$ subblocks, implying earlier
termination for high values of~$|M|$ and~$r$ and explaining the slightly longer
run time for $M=(2,\vee,0)$ on small~$r$. We note in summary that WTAs with well
over $15$\hpnote{} million edges are processed in
less than five minutes, and in fact the run time of minimization is of the same
order of magnitude as that of input parsing.
\newcommand{\boolmon}{$(2,\vee,0)$}%
\newcommand{\Zmax}{$(\N,\max,0)$}%
\newcommand{\WordOr}{$(\Powf(64),\cup,\emptyset)$}%
\begin{table}
\caption{Processing times for partition refinement on maximal weighted tree
    automata (i.e.~coalgebras for $M\times M^{(\Sigma(-))}$) in 16~GB of memory with $n$ states and 50 transitions per state,
    leading to $n'$ states and $m$ edges in total. The column `Init' provides
    the time needed to compute the initial partition $P_1$ (on all $n'$ states), and `Refine' the
    time to compute the final partitions $P_f'$ (on $n'$ states) and $P_f$ (omitting intermediate states).}\label{tab:extended}
  \medskip
  \def\arraystretch{1.3}
  \begin{tabular}{@{}l@{\hspace{2mm}}lcr@{~~}r@{~~}r@{~~}r@{~~}cr@{~~}r@{~~}r@{~~}cr@{~~}r@{~~}r@{}}
\toprule
    \multicolumn{2}{c}{Functor Parameters}
    &&\multicolumn{4}{c}{Input File} 
    &&\multicolumn{3}{c}{Partition Sizes} 
    &&\multicolumn{3}{c}{Time (s) to} 
    \\
    \cline{1-2}
    \cline{4-7}
    \cline{9-11}
    \cline{13-15}
    Monoid $M$ & $\Sigma X =$ %
    && \multicolumn{1}{c}{$n$} & \multicolumn{1}{c}{$n'$}
     & \multicolumn{1}{c}{$m$} & \multicolumn{1}{c}{Size}
    && \multicolumn{1}{c}{$P_1'$} & \multicolumn{1}{c}{$P_f'$} & \multicolumn{1}{c}{$P_f$}
    && \multicolumn{1}{c}{Parse} & \multicolumn{1}{c}{Init} & Refine \\
\midrule
\boolmon & $4\times X$   && 154863 & 7898013 & 15486300 & 101 MB && 6 & 774313  & 154863 && 53 & 36 & 183 \\
         & $4\times X^2$ && 138000 & 7038000 & 20700000 & 143 MB && 6 & 7037670 & 138000 && 54 & 31 & 311 \\
         & $4\times X^3$ && 134207 & 6844557 & 26841400 & 191 MB && 6 & 6844557 & 134207 && 62 & 27 & 285 \\
         & $4\times X^4$ && 92491  & 4717041 & 23122750 & 163 MB && 6 & 4717041 & 92491  && 51 & 28 & 175 \\
         & $4\times X^5$ && 86852  & 4429452 & 26055600 & 186 MB && 6 & 4429452 & 86852  && 54 & 26 & 165 \\
\midrule
\Zmax    & $4\times X$   && 156913 & 8002563 & 15691300 & 131 MB && 437 & 784564  & 156913 && 62 & 31 & 93  \\
         & $4\times X^2$ && 118084 & 6022284 & 17712600 & 143 MB && 416 & 6021960 & 118084 && 58 & 30 & 119 \\
         & $4\times X^3$ && 100799 & 5140749 & 20159800 & 158 MB && 414 & 5140749 & 100799 && 57 & 28 & 105 \\
         & $4\times X^4$ && 92879  & 4736829 & 23219750 & 181 MB && 409 & 4736829 & 92879  && 60 & 28 & 105 \\
         & $4\times X^5$ && 94451  & 4817001 & 28335300 & 219 MB && 417 & 4817001 & 94451  && 63 & 29 & 108 \\
\midrule
\WordOr  & $4\times X$   && 152107 & 7757457 & 15210700 & 141 MB && 54 & 760534  & 152107 && 65 & 27 & 149 \\
         & $4\times X^2$ && 134082 & 6838182 & 20112300 & 176 MB && 54 & 6837891 & 134082 && 66 & 29 & 229 \\
         & $4\times X^3$ && 94425  & 4815675 & 18885000 & 157 MB && 54 & 4815675 & 94425  && 57 & 25 & 185 \\
         & $4\times X^4$ && 83431  & 4254981 & 20857750 & 170 MB && 54 & 4254981 & 83431  && 56 & 24 & 175 \\
         & $4\times X^5$ && 92615  & 4723365 & 27784500 & 223 MB && 54 & 4723365 & 92615  && 64 & 21 & 194 \\
  \bottomrule
\end{tabular}
\end{table}%
Since the publication of the conference paper~\cite{coparFM19}, we have
optimized the memory consumption in \copar, especially in the refinement
interface for the functor $\Sigma$. With the optimizations, \copar can now
handle coalgebras with 20\% more states within the same memory limit of 16GB of RAM
for the cases with $\Sigma X = 4 \times X$ and even 75\% more states for
the cases with $\Sigma X = 4 \times X^5$.

\subsection{Weighted Tree Automata from Grammars}
One widespread practical use of weighted tree automata is to learn natural language grammars~\cite{PetrovEA06Learning,PetrovKlein07Improved}. 
Methods of this type have been implemented in the \berkeleyparser{}
project\footnote{Available at
  \url{https://github.com/slavpetrov/berkeleyparser}}. The project also makes
six language grammars available, which we use to
demonstrate that \copar is capable of handling inputs that arise in practice.

We have parsed\footnote{See
  \url{https://git8.cs.fau.de/software/copar-benchmarks/-/tree/master/berkeleyparser}
for details} the weighted tree automata given in the
\berkeleyparser-specific format, obtaining coalgebras for the functor
\[
  FX= \R^{(\Sigma X)}
\]
where $\Sigma$ is (the polynomial functor for) a ranked alphabet containing, in the case at hand, only
operation symbols of arities~1 and~2. The ranked alphabet is implicit in the
original grammar file. In the grammar file, every state (i.e.~symbol of the grammar) is of the shape $S_i$ 
where $i\in \N$ is an index and~$S$ indicates that there are operation symbols
$\arity{S}{1}$  and $\arity{S}{2}$ in $\Sigma$. The rules in the grammar files
are of the shape
\[
  S_i \xrightarrow{w}  T_j R_k
  \qquad\text{or}\qquad
  S_i \xrightarrow{w}  T_j
\]
where $S_i$, $T_j$, $R_k$ are states and $w\in [0,1]$ is the weight of the
transition. In the induced coalgebra $\mu\colon X\to (\R,+,0)^{\Sigma X}$, the
state set $X$ is the set of symbols, and the above rules correspond to the transitions
\[
  \mu(S_i)(\arity{S}{2}(T_j, R_k)) = w
  \qquad\text{resp.}\qquad
  \mu(S_i)(\arity{S}{1}(T_j)) = w.
\]
Hence, by definition of the \berkeleyparser-format, every symbol (i.e.~state)
mentions at most two different operation symbols in its outgoing transitions.
\begin{table}
  \caption{Performance on weighted tree automata from the \berkeleyparser
    project, after transforming them to $\R^{(\Sigma(-))}$-coalgebras. The
    column `Init' provides the time needed to compute the initial partition
    $P_1'$ (on all $n'$ states) and `Refine' the time to compute the final
    partitions $P_f'$ (on $n'$ states) and $P_f$ (omitting intermediate states).}
  \medskip
  \def\arraystretch{1.3}
  \label{tab:bench-berkeley}
\begin{tabular}{@{}l@{}c@{}rrr@{}crrr@{}c@{}rrr@{}}
    \toprule
    Filename
    &\hspace*{2mm}& \multicolumn{3}{c}{Input}
    &\hspace*{2mm}& \multicolumn{3}{@{}c@{}}{Partion Sizes}
    &\hspace*{3mm}& \multicolumn{3}{c}{Time (s) to} \\
    \cline{3-5}
    \cline{7-9}
    \cline{11-13}
    && \multicolumn{1}{r}{$n$}
    & \multicolumn{1}{r}{$n'$}
    & \multicolumn{1}{r@{}}{$m$}
    && $P_1'$
    & $P_f'$
    & $P_f$
    && Parse
    & Init
    & Refine
    \\
    \midrule
\texttt{arb\_sm5}
                &
                & 1026
                & 946509
                & 2788134
                &
                & 315
                & 26427
                & 404
                &
                & 9.41
                & 2.47
                & 3.63
            \\
\texttt{bul\_sm5}
                &
                & 2545
                & 834777
                & 2173611
                &
                & 743
                & 37660
                & 1327
                &
                & 8.80
                & 2.07
                & 2.57
            \\
\texttt{chn\_sm5}
                &
                & 945
                & 1110445
                & 3236715
                &
                & 426
                & 31386
                & 579
                &
                & 11.31
                & 3.28
                & 4.43
            \\
\texttt{eng\_sm6}
                &
                & 1133
                & 1843351
                & 5410006
                &
                & 449
                & 30692
                & 548
                &
                & 19.71
                & 5.25
                & 7.59
            \\
\texttt{fra\_sm5}
                &
                & 737
                & 451171
                & 1343591
                &
                & 178
                & 18770
                & 319
                &
                & 4.21
                & 0.98
                & 1.57
            \\
\texttt{ger\_sm5}
                &
                & 986
                & 617762
                & 1848751
                &
                & 273
                & 20324
                & 484
                &
                & 6.22
                & 1.48
                & 2.03
            \\
\bottomrule
    \end{tabular}

\end{table}
\autoref{tab:bench-berkeley} shows the performance of \copar on the input files
obtained. Most of the running time is spent during parsing of our generic input
format (\autoref{sec:parsing}) whereas the actual partition refinement runs in
under 13 seconds on each of the files. The number of states mentioned explicitly
in the files is denoted by $n$, and~$n'$ denotes the total number of states
obtained after introducing intermediate states via our modularity mechanism
(\autoref{sec:des}).  As one can tell from the partition sizes $P_1'$ and
$P_f'$, a considerable number of refinement steps is needed before the final partition is reached.

\subsection{Benchmarks for PRISM Models}
\newcommand{\benchparam}[1]{\scriptsize \textnormal{(}#1\textnormal{)}}
\begin{table}
  \caption{Performance on PRISM benchmarks}
  \medskip
  \noshowkeys\label{tab:bench-prism}
  \def\arraystretch{1.2}
  \begin{tabular}{@{}llrr@{}c@{}r@{\hspace{0.5em}}rr@{}c@{}rr@{}}
    \toprule
    PRISM Model &Functor& \multicolumn{2}{c}{Input}
                &\hspace*{2mm}& \multicolumn{3}{c}{Time (s) to}
                &\hspace*{2mm}& \multicolumn{2}{c}{Time (s) of} \\
    \cline{3-4}
    \cline{6-8}
    \cline{10-11}
                        && \multicolumn{1}{l}{States}
                        & \multicolumn{1}{l}{Edges}
                        &
                        & Parse
                        & Init
                        & Refine
                        &
                        & Valmari
                        & mCRL2 \\
    \midrule
    \tt fms \benchparam{n=4} &$\R^{(-)}$& 35910   & 237120  && 0.48  & 0.12 & 0.16  && 0.21 & \multicolumn{1}{c}{--} \\
\tt fms \benchparam{n=5}     &$\R^{(-)}$& 152712  & 1111482 && 2.46  & 0.68 & 1.1   && 1.21 & \multicolumn{1}{c}{--} \\
\tt fms \benchparam{n=6}     &$\R^{(-)}$& 537768  & 4205670 && 9.94  & 2.91 & 5.56  && 5.84 & \multicolumn{1}{c}{--} \\[1mm]
    \midrule
    \multirow{1}{*}{%
\tt \makecell[l]{wlan2\_collide\\[-1mm]\benchparam{COL=2,TRANS\_TIME\_MAX=10}}}                                           &$\N\times\Pow(\N\times \Dist (-))$& 65718   & 94452   && 0.5   & 0.3  & 0.58  && 0.12 & 0.41                   \\[2mm]
\multirow{1}{*}{\tt \makecell[l]{wlan0\_time\_bounded\\[-1mm]{\scriptsize \textnormal{(}TRANS\_TIME\_MAX=10,DEADLINE=100\textnormal{)}}}} &$\N\times\Pow(\N\times \Dist (-))$& 582327  & 771088  && 5.19  & 3.13 & 5.5   && 0.88 & 3.18                   \\[2mm]
\multirow{1}{*}{\tt \makecell[l]{wlan1\_time\_bounded\\[-1mm]\benchparam{TRANS\_TIME\_MAX=10,DEADLINE=100}}}                              &$\N\times\Pow(\N\times \Dist (-))$& 1408676 & 1963522 && 13.37 & 6.18 & 16.18 && 2.44 & 8.44                   \\[2mm]
    \bottomrule
  \end{tabular}
\end{table}

In order to see how \copar performs on models of other system types that arise in practice, we have
taken two kinds of models from the benchmark
suite~\cite{DBLP:conf/qest/KwiatkowskaNP12} of the probabilistic model checker
PRISM~\cite{KNP11}. We derived coalgebras for the functors
\begin{itemize}
\item $FX = \R^{(X)}$ from continuous time Markov chains (CTMC), and
\item $FX = \N\times \Pow(\N\times(\Dist X))$ from Markov decision processes (MDP).
\end{itemize}
This translation deliberately ignores the variable valuations present in the
original benchmark models to avoid situations where all states are already
distinguished after the first refinement step. For MDPs, the translation instead
generates a coarse initial partition for each model (the outer $\N\times (-)$).
For the CTMCs considered, the functor $\R^{(-)}$ is already sufficient since the
initial partition distinguishes states by the accumulated weight of their
outgoing transitions.

Like in the case of WTAs, the functor for MDPs is a composite of several basic
functors and thus requires use of the construction described in
\autoref{sec:des}. Two of the benchmarks are shown in \autoref{tab:bench-prism}
with different parameters, resulting in three differently sized coalgebras
each. The \emph{fms} family of systems model a flexible manufacturing
system~\cite{DBLP:journals/pe/CiardoT93} as CTMCs (without initial partition),
and we minimize them under the usual weighted bisimilarity, i.e.~as
$\R^{(-)}$-coalgebras. The \emph{wlan}
benchmarks~\cite{DBLP:conf/papm/KwiatkowskaNS02} model various aspects of the
IEEE 802.11 Wireless LAN protocol as MDPs.

\autoref{tab:bench-prism} also includes the total run time of two additional
partition refinement tools: Valmari's C++ implementation\footnote{Available at
  \url{https://git8.cs.fau.de/hpd/mdpmin-valmari}} of algorithms described by Valmari and Franceschinis~\cite{ValmariF10,valmari2010simple},
which can minimize MDPs as well as CTMCs, and the tool \texttt{ltspbisim} from
the mCRL2 toolset\footnote{One should note that mCRL2 offers a whole suite of
  reasoning services besides partition refinement.}~\cite{BunteEA19} version
201808.0, which implements a recently discovered refinement algorithm for
MDPs~\cite{GrooteEA18} (but does not support CTMCs directly, hence there is no
data in the first three lines).

The results in \autoref{tab:bench-prism} show that refinement for the \emph{fms}
benchmarks is faster than for the respective \emph{wlan} ones, even though the
first group has more edges. This is due to~(a) the fact that the functor
for MDPs is more complex and thus introduces more indirection into our
algorithms, as explained in \autoref{sec:des}, and (b)~that our optimization for
one-element blocks fires much more often for
\emph{fms}.

It is also apparent that \copar{} is slower than both of the other tools in our
comparison for the presented examples in \autoref{tab:bench-prism}: CoPaR takes
slightly more than 30 seconds whereas Valmari's optimized implementation only
takes 2.44. To some extent, this performance difference can be attributed to the
fact that our implementation is written in Haskell and the other tools in C++.
In addition, \copar{} incurs a certain amount of overhead for genericity and
modularity. Moreover, CoPaR's input format as described in
\autoref{sec:parsing} is much more complex to parse than Valmari's format, which
is essentially a whitespace separated list of integers

\section{Conclusion and Future Work}\label{section:conclusion}
\label{sec:conc}

\takeout{}%

\noindent We have instantiated a generic and efficient partition refinement
algorithm that we introduced in previous work \cite{concurSpecialIssue} to
weighted (tree) automata, and we have refined the generic complexity analysis of
the algorithm to cover this case. Moreover, we have described an implementation
of the generic algorithm in the form of the tool \copar, which supports the
modular combination of basic system types without requiring any additional
implementation effort, and allows for easy incorporation of new basic system
types by implementing a generic refinement interface.  \copar{} is currently
concerned entirely with partition refinement, and does not implement other
algorithmic tasks (such as simulation, visualization, or model checking), which
for specific system types such as labelled transition
systems or Markov chains are covered by existing well-developed tool suites
(\autoref{sec:bench}). The salient feature of \copar{} is its genericity, which
allows for instantiation of the partition refinement algorithm to new system
types with minimal effort.

In future work, we will further broaden the range of system types that our
algorithm and tool can accommodate, and provide support for base categories
beyond sets, e.g.~nominal sets, which underlie nominal
automata~\cite{BojanczykEA14,SchroderEA17}, or algebraic
categories~\cite{ZoltanMaletti11,generalizeddeterminization}.

Concerning genericity, there is an orthogonal approach by Ranzato and
Tapparo~\cite{RanzatoT08}, which is generic over \emph{notions of process
  equivalence} but fixes the system type to standard labelled transition
systems; see also~\cite{GrooteEA17}. Similarly, Blom and
Orzan~\cite{BlomOrzan03,BlomOrzan05} present \emph{signature refinement}, which
covers, e.g.~strong and branching bisimulation as well as Markov chain lumping,
but requires adapting the algorithm for each instance. These algorithms have
also been improved using symbolic techniques (e.g.~\cite{vDvdP18}). Moreover,
many of the mentioned approaches and 
others~\cite{BergaminiEA05,BlomOrzan03,BlomOrzan05,GaravelHermanns02,vDvdP18}
focus on parallelization. We will explore in future work whether symbolic and
distributed methods can be lifted to coalgebraic generality.

In a recent alternative approach to bisimilarity minimization called
\emph{partition aggregation}~\cite{BjorklundCleophas2020}, the behavioural
equivalence relation is approximated from below, rather than from above as in
partition refinement. Partition aggregation has worse run-time complexity than
partition refinement on the global task of minimizing entire (reachable)
systems, but it can be executed partially and thus may be more efficient on the
local task of checking equivalence of two given states in a labelled transition
system. In future work, we aim at a coalgebraic generalization of partition
aggregation, which might help circumvent the bottleneck of (linear) memory
consumption that is intrinsic to partition
refinement~\cite{valmari2010simple,Valmari09}.

\bibliographystyle{alpha}
\bibliography{refs}

\newcommand{\etalchar}[1]{$^{#1}$}
\begin{thebibliography}{HBMM09}

\bibitem[Ada93]{Adams93}
Stephen Adams.
\newblock Efficient sets - {A} balancing act.
\newblock {\em J.\ Funct.\ Program.}, 3(4):553--561, 1993.

\bibitem[Awo10]{Awodey10}
Steve Awodey.
\newblock {\em Category Theory}, volume~52 of {\em Oxford Logic Guides}.
\newblock Oxford University Press, 2 edition, 2010.

\bibitem[BBG17]{BerkholzBG17}
Christoph Berkholz, Paul~S. Bonsma, and Martin Grohe.
\newblock Tight lower and upper bounds for the complexity of canonical colour
  refinement.
\newblock {\em Theory Comput. Syst.}, 60(4):581--614, 2017.

\bibitem[BC20]{BjorklundCleophas2020}
Johanna Bj\"{o}rklund and Loek Cleophas.
\newblock Aggregation-based minimization of finite state automata.
\newblock {\em Acta Informatica}, January 2020.

\bibitem[BDJM05]{BergaminiEA05}
Damien Bergamini, Nicolas Descoubes, Christophe Joubert, and Radu Mateescu.
\newblock {BISIMULATOR:} {A} modular tool for on-the-fly equivalence checking.
\newblock In {\em Tools and Algorithms for the Construction and Analysis of
  Systems, {TACAS} 2005}, volume 3440 of {\em LNCS}, pages 581--585. Springer,
  2005.

\bibitem[BEM00]{BaierEM00}
Christel Baier, Bettina Engelen, and Mila Majster{-}Cederbaum.
\newblock Deciding bisimilarity and similarity for probabilistic processes.
\newblock {\em J.\ Comput.\ Syst.\ Sci.}, 60:187--231, 2000.

\bibitem[BGK{\etalchar{+}}19]{BunteEA19}
Olav Bunte, Jan~Friso Groote, Jeroen J.~A. Keiren, Maurice Laveaux, Thomas
  Neele, Erik~P. de~Vink, Wieger Wesselink, Anton Wijs, and Tim A.~C. Willemse.
\newblock The {mCRL2} toolset for analysing concurrent systems - improvements
  in expressivity and usability.
\newblock In {\em Tools and Algorithms for the Construction and Analysis of
  Systems, {TACAS} 2019}, pages 21--39, 2019.

\bibitem[BKL14]{BojanczykEA14}
Miko{\l}aj Boja{\'{n}}czyk, Bartek Klin, and Slawomir Lasota.
\newblock Automata theory in nominal sets.
\newblock {\em Log.\ Methods Comput.\ Sci.}, 10(3), 2014.

\bibitem[BO03]{BlomOrzan03}
Stefan Blom and Simona Orzan.
\newblock Distributed branching bisimulation reduction of state spaces.
\newblock In {\em Parallel and Distributed Model Checking, PDMC 2003},
  volume~89 of {\em ENTCS}, pages 99--113. Elsevier, 2003.

\bibitem[BO05]{BlomOrzan05}
Stefan Blom and Simona Orzan.
\newblock A distributed algorihm for strong bisimulation reduction of state
  spaces.
\newblock {\em J.\ Softw.\ Tools Technol.\ Transfer}, 7(1):74--86, 2005.

\bibitem[BSdV03]{BARTELS200357}
Falk Bartels, Ana Sokolova, and Erik de~Vink.
\newblock A hierarchy of probabilistic system types.
\newblock In {\em Coagebraic Methods in Computer Science, CMCS 2003}, volume~82
  of {\em ENTCS}, pages 57 -- 75. Elsevier, 2003.

\bibitem[Buc08]{Buchholz08}
Peter Buchholz.
\newblock Bisimulation relations for weighted automata.
\newblock {\em Theor.\ Comput.\ Sci.}, 393:109--123, 2008.

\bibitem[CLR90]{CormenEA}
Thomas Cormen, Charles Leiserson, and Ronald Rivest.
\newblock {\em Introduction to Algorithms}.
\newblock MIT Press, 1990.

\bibitem[CT93]{DBLP:journals/pe/CiardoT93}
Gianfranco Ciardo and Kishor~S. Trivedi.
\newblock A decomposition approach for stochastic reward net models.
\newblock {\em Perform. Evaluation}, 18(1):37--59, 1993.

\bibitem[Dei19]{Deifel18}
Hans-Peter Deifel.
\newblock Implementation and evaluation of efficient partition refinement
  algorithms.
\newblock Master's thesis, Friedrich-Alexander Universität Erlangen-Nürnberg,
  2019.
\newblock \url{https://hpdeifel.de/master-thesis-deifel.pdf}.

\bibitem[DHS03]{DerisaviEA03}
Salem Derisavi, Holger Hermanns, and William Sanders.
\newblock Optimal state-space lumping in {M}arkov chains.
\newblock {\em Inf.\ Process.\ Lett.}, 87(6):309--315, 2003.

\bibitem[DMSW17]{concur2017}
Ulrich Dorsch, Stefan Milius, Lutz Schr{\"o}der, and Thorsten Wi{\ss}mann.
\newblock {Efficient Coalgebraic Partition Refinement}.
\newblock In {\em Concurrency Theory, CONCUR 2017}, volume~85 of {\em LIPIcs},
  pages 32:1--32:16. Schloss Dagstuhl -- Leibniz-Zentrum für Informatik, 2017.

\bibitem[DMSW19]{coparFM19}
Hans-Peter Deifel, Stefan Milius, Lutz Schr{\"o}der, and Thorsten Wi{\ss}mann.
\newblock Generic partition refinement and weighted tree automata.
\newblock In Maurice~H. ter Beek, Annabelle McIver, and Jos{\'e}~N. Oliveira,
  editors, {\em Formal Methods -- The Next 30 Years}, pages 280--297, Cham, 10
  2019. Springer International Publishing.

\bibitem[DPP04]{DovierEA04}
Agostino Dovier, Carla Piazza, and Alberto Policriti.
\newblock An efficient algorithm for computing bisimulation equivalence.
\newblock {\em Theor. Comput. Sci.}, 311(1-3):221--256, 2004.

\bibitem[EM11]{ZoltanMaletti11}
Zoltan Esik and Andreas Maletti.
\newblock The category of simulations for weighted tree automata.
\newblock {\em Int. J. Found. Comput. Sci.}, 22:1845--1859, 12 2011.

\bibitem[GH02]{GaravelHermanns02}
Hubert Garavel and Holger Hermanns.
\newblock On combining functional verification and performance evaluation using
  {CADP}.
\newblock In {\em Formal Methods Europe, {FME} 2002}, volume 2391 of {\em
  LNCS}, pages 410--429. Springer, 2002.

\bibitem[GJKW17]{GrooteEA17}
Jan~Friso Groote, David~N. Jansen, Jeroen~J.A. Keiren, and Anton Wijs.
\newblock An \emph{O}(\emph{m}log\emph{n}) algorithm for computing stuttering
  equivalence and branching bisimulation.
\newblock {\em {ACM} Trans. Comput. Log.}, 18(2):13:1--13:34, 2017.

\bibitem[Gri73]{Gries1973}
David Gries.
\newblock Describing an algorithm by {H}opcroft.
\newblock {\em Acta Informatica}, 2:97--109, 1973.

\bibitem[GVdV18]{GrooteEA18}
Jan~Friso Groote, Jao~Rivera Verduzco, and Erik~P. de~Vink.
\newblock An efficient algorithm to determine probabilistic bisimulation.
\newblock {\em Algorithms}, 11(9):131, 2018.

\bibitem[HBMM07]{HoegbergEA07}
Johanna H{\"{o}}gberg~(Björklund), Andreas Maletti, and Jonathan May.
\newblock Bisimulation minimisation for weighted tree automata.
\newblock In {\em Developments in Language Theory, {DLT} 2007}, volume 4588 of
  {\em LNCS}, pages 229--241. Springer, 2007.

\bibitem[HBMM09]{HoegbergEA09}
Johanna H\"{o}gberg~(Björklund), Andreas Maletti, and Jonathan May.
\newblock Backward and forward bisimulation minimization of tree automata.
\newblock {\em Theor.\ Comput.\ Sci.}, 410:3539--3552, 2009.

\bibitem[Hop71]{Hopcroft71}
John Hopcroft.
\newblock An $n \log n$ algorithm for minimizing states in a finite automaton.
\newblock In {\em Theory of Machines and Computations}, pages 189--196.
  Academic Press, 1971.

\bibitem[HT92]{HuynhTian92}
Dung Huynh and Lu~Tian.
\newblock On some equivalence relations for probabilistic processes.
\newblock {\em Fund.\ Inform.}, 17:211--234, 1992.

\bibitem[KNP11]{KNP11}
Marta Kwiatkowska, Gethin Norman, and David Parker.
\newblock {PRISM} 4.0: Verification of probabilistic real-time systems.
\newblock In {\em Computer Aided Verification, CAV 2011}, volume 6806 of {\em
  LNCS}, pages 585--591. Springer, 2011.

\bibitem[KNP12]{DBLP:conf/qest/KwiatkowskaNP12}
Marta~Z. Kwiatkowska, Gethin Norman, and David Parker.
\newblock The {PRISM} benchmark suite.
\newblock In {\em Ninth International Conference on Quantitative Evaluation of
  Systems, {QEST} 2012, London, United Kingdom, September 17-20, 2012}, pages
  203--204. {IEEE} Computer Society, 2012.

\bibitem[KNS02]{DBLP:conf/papm/KwiatkowskaNS02}
Marta~Z. Kwiatkowska, Gethin Norman, and Jeremy Sproston.
\newblock Probabilistic model checking of the {IEEE} 802.11 wireless local area
  network protocol.
\newblock In Holger Hermanns and Roberto Segala, editors, {\em Process Algebra
  and Probabilistic Methods, Performance Modeling and Verification, Second
  Joint International Workshop {PAPM-PROBMIV} 2002, Copenhagen, Denmark, July
  25-26, 2002, Proceedings}, volume 2399 of {\em Lecture Notes in Computer
  Science}, pages 169--187. Springer, 2002.

\bibitem[Knu01]{Knuutila2001}
Timo Knuutila.
\newblock Re-describing an algorithm by {H}opcroft.
\newblock {\em Theor.\ Comput.\ Sci.}, 250:333--363, 2001.

\bibitem[KS90]{KanellakisS90}
Paris Kanellakis and Scott Smolka.
\newblock {CCS} expressions, finite state processes, and three problems of
  equivalence.
\newblock {\em Inf. Comput.}, 86(1):43--68, 1990.

\bibitem[KS13]{KlinS13}
Bartek Klin and Vladimiro Sassone.
\newblock Structural operational semantics for stochastic and weighted
  transition systems.
\newblock {\em Inf.~Comput.}, 227:58--83, 2013.

\bibitem[LP94]{DBLP:conf/pldi/LaunchburyJ94}
John Launchbury and Simon~L. {Peyton Jones}.
\newblock Lazy functional state threads.
\newblock In Vivek Sarkar, Barbara~G. Ryder, and Mary~Lou Soffa, editors, {\em
  Proceedings of the {ACM} SIGPLAN'94 Conference on Programming Language Design
  and Implementation (PLDI), Orlando, Florida, USA, June 20-24, 1994}, pages
  24--35. {ACM}, 1994.

\bibitem[Mil80]{Milner80}
Robin Milner.
\newblock {\em A Calculus of Communicating Systems}, volume~92 of {\em LNCS}.
\newblock Springer, 1980.

\bibitem[MK06]{MayKnight06}
Jonathan May and Kevin Knight.
\newblock Tiburon: A weighted tree automata toolkit.
\newblock In Oscar~H. Ibarra and Hsu-Chun Yen, editors, {\em Implementation and
  Application of Automata}, pages 102--113, Berlin, Heidelberg, 2006. Springer
  Berlin Heidelberg.

\bibitem[Par81]{Park81}
David Park.
\newblock Concurrency and automata on infinite sequences.
\newblock In {\em Theoretical Computer Science, 5th GI-Conference}, volume 104
  of {\em LNCS}, pages 167--183. Springer, 1981.

\bibitem[PBTK06]{PetrovEA06Learning}
Slav Petrov, Leon Barrett, Romain Thibaux, and Dan Klein.
\newblock Learning accurate, compact, and interpretable tree annotation.
\newblock In {\em Proceedings of the 21st International Conference on
  Computational Linguistics and 44th Annual Meeting of the Association for
  Computational Linguistics}, pages 433--440, Sydney, Australia, July 2006.
  Association for Computational Linguistics.

\bibitem[PK07]{PetrovKlein07Improved}
Slav Petrov and Dan Klein.
\newblock Improved inference for unlexicalized parsing.
\newblock In {\em Human Language Technologies 2007: The Conference of the North
  {A}merican Chapter of the Association for Computational Linguistics;
  Proceedings of the Main Conference}, pages 404--411, Rochester, New York,
  April 2007. Association for Computational Linguistics.

\bibitem[PT87]{PaigeTarjan87}
Robert Paige and Robert Tarjan.
\newblock Three partition refinement algorithms.
\newblock {\em SIAM J.~Comput.}, 16(6):973--989, 1987.

\bibitem[RT08]{RanzatoT08}
Francesco Ranzato and Francesco Tapparo.
\newblock Generalizing the {P}aige-{T}arjan algorithm by abstract
  interpretation.
\newblock {\em Inf.\ Comput.}, 206:620--651, 2008.

\bibitem[Rut00]{Rutten00}
Jan Rutten.
\newblock Universal coalgebra: a theory of systems.
\newblock {\em Theor.\ Comput.\ Sci.}, 249:3--80, 2000.

\bibitem[SBBR13]{generalizeddeterminization}
Alexandra Silva, Filippo Bonchi, Marcello~M. Bonsangue, and Jan J. M.~M.
  Rutten.
\newblock Generalizing determinization from automata to coalgebras.
\newblock {\em Logical Methods in Computer Science}, 9(1), 2013.

\bibitem[Seg95]{Segala95}
Roberto Segala.
\newblock {\em Modelling and Verification of Randomized Distributed Real-Time
  Systems}.
\newblock PhD thesis, MIT, 1995.

\bibitem[SKMW17]{SchroderEA17}
Lutz Schr{\"{o}}der, Dexter Kozen, Stefan Milius, and Thorsten Wi{\ss}mann.
\newblock Nominal automata with name binding.
\newblock In {\em Foundations of Software Science and Computation Structures,
  {FOSSACS} 2017}, volume 10203 of {\em LNCS}, pages 124--142, 2017.

\bibitem[Val09]{Valmari09}
Antti Valmari.
\newblock Bisimilarity minimization in {$\CO(m \log n)$} time.
\newblock In {\em Applications and Theory of Petri Nets, {PETRI} {NETS} 2009},
  volume 5606 of {\em LNCS}, pages 123--142. Springer, 2009.

\bibitem[Val10]{valmari2010simple}
Antti Valmari.
\newblock Simple bisimilarity minimization in {$\CO(m \log n)$} time.
\newblock {\em Fund.\ Inform.}, 105(3):319--339, 2010.

\bibitem[vDvdP18]{vDvdP18}
Tom van Dijk and Jaco van~de Pol.
\newblock Multi-core symbolic bisimulation minimization.
\newblock {\em J.~Softw.~Tools~Technol.~Transfer}, 20(2):157--177, 2018.

\bibitem[VF10]{ValmariF10}
Antti Valmari and Giuliana Franceschinis.
\newblock Simple {$\CO(m\log n)$} time {M}arkov chain lumping.
\newblock In {\em Tools and Algorithms for the Construction and Analysis of
  Systems, TACAS 2010}, volume 6015 of {\em LNCS}, pages 38--52. Springer,
  2010.

\bibitem[vG01]{vanglabbeek2001linear}
R.~van Glabbeek.
\newblock The linear time -- branching time spectrum {I}; the semantics of
  concrete, sequential processes.
\newblock In J.~Bergstra, A.~Ponse, and S.~Smolka, editors, {\em Handbook of
  Process Algebra}, pages 3--99. Elsevier, 2001.

\bibitem[WDMS20]{concurSpecialIssue}
Thorsten Wißmann, Ulrich Dorsch, Stefan Milius, and Lutz Schröder.
\newblock {Efficient and Modular Coalgebraic Partition Refinement}.
\newblock {\em {Logical Methods in Computer Science}}, {Volume 16, Issue 1},
  January 2020.

\end{thebibliography}

\label{lastpage}
\end{document}